\documentclass[sigplan,10pt]{acmart}

\usepackage{header}
\usepackage{code}
\usepackage{defs}
\usepackage{enumitem}
         
\acmConference[PL'17]{ACM SIGPLAN Conference on Programming Languages}{January 01--03, 2017}{New York, NY, USA}
\acmYear{2017}
\acmISBN{} 
\acmDOI{} 
\startPage{1}

\setcopyright{none}             

\begin{document}

\newif\ifsupp
\supptrue

\title{Verifying Semantic Conflict-Freedom in Three-Way Program Merges}

\author{Marcelo Sousa}
\affiliation{
  \institution{University of Oxford}
  \country{United Kingdom}
}
\email{marcelo.sousa@cs.ox.ac.uk}
\author{Isil Dillig}
\affiliation{
  \institution{University of Texas at Austin}
  \country{United States}
}
\email{isil@cs.utexas.edu}
\author{Shuvendu Lahiri}
\affiliation{
  \institution{Microsoft Research}
  \country{United States}
}
\email{shuvendu.lahiri@microsoft.com}

\begin{abstract}
Even though many programmers rely on 3-way 
merge tools to integrate changes from different branches, such tools 
 can introduce subtle bugs in the integration process.
This paper aims to mitigate this problem by defining a semantic notion of 
\emph{conflict-freedom}, which ensures that the merged program does not 
introduce new unwanted behaviors. 
We also show how to verify this property using a novel, compositional algorithm 
that combines
lightweight dependence analysis for shared program fragments and precise relational 
reasoning for the modifications.
%
%
%
%
We evaluate our tool called \toolname\ on 52 real-world merge scenarios obtained from Github and compare the results against a textual merge tool. The experimental results demonstrate the benefits of our approach over syntactic conflict-freedom and indicate that \toolname\ is both precise and practical.
%
%
\end{abstract}

\maketitle
\section{Introduction}
\label{sec:intro}

Developers who edit different branches of a source code repository
rely on 3-way merge tools (like {\tt \small git-merge} or {\tt \small kdiff3})
to automatically merge their changes.
Since the vast majority of these tools are oblivious
to program semantics and resolve conflicts using syntactic criteria,
they may introduce bugs in the merge process.
For example, many people speculate that Apple's infamous {\tt \small goto fail} SSL bug
was introduced  due to an erroneous program merge~\cite{goto-bug1,goto-bug2,goto-bug3}.

To see how bugs may be introduced in the merge process, consider the
simple \emph{base program} shown in Figure~\ref{fig:ex-basic} together
with its two \emph{variants} $A$ and $B$.\footnote{The example is inspired by the Apple SSL bug that resulted from duplicate goto statements.}
Here, both $A$ and $B$ modify the original program by incrementing
variable $\mytt{x}$ by $1$. For instance, such a situation may arise
in practice when two independent developers simultaneously fix the
same bug in different locations of the original program.
Since both variants effectively make the same change, the correct merge
should be either $A$ or $B$.
However, running a 3-way merge tool (in this case, {\tt \small kdiff3}) on
these programs succeeds without any warnings and generates the incorrect merge shown on the right hand
side of Figure~\ref{fig:ex-basic}. 
Since this program is clearly different than what either developer 
intended, we see that a bug was introduced during the merge.


This paper takes a step towards eliminating bugs that arise
due to 3-way program merges by automatically verifying \emph{semantic conflict-freedom}, a notion inspired by 
earlier work on program integration~\cite{hpr, Yang90}. To motivate what we mean by semantic conflict-freedom, 
consider a base program $P$, two variants $A, B$, and a merge candidate $M$. Intuitively, semantic conflict freedom requires that, if variant $A$ (resp. $B$) disagrees with $P$ on the value of some program variable $v$, then the merge candidate $M$ should agree with $A$ (resp. $B$) on the value of $v$.  In addition to ensuring that the merge candidate  does not introduce new behavior that is not present in either of the variants, conflict freedom also ensures that variants $A$ and $B$ do not make changes that are semantically incompatible with each other.

\begin{figure}[!t]
\begin{center}
\includegraphics[scale=0.32]{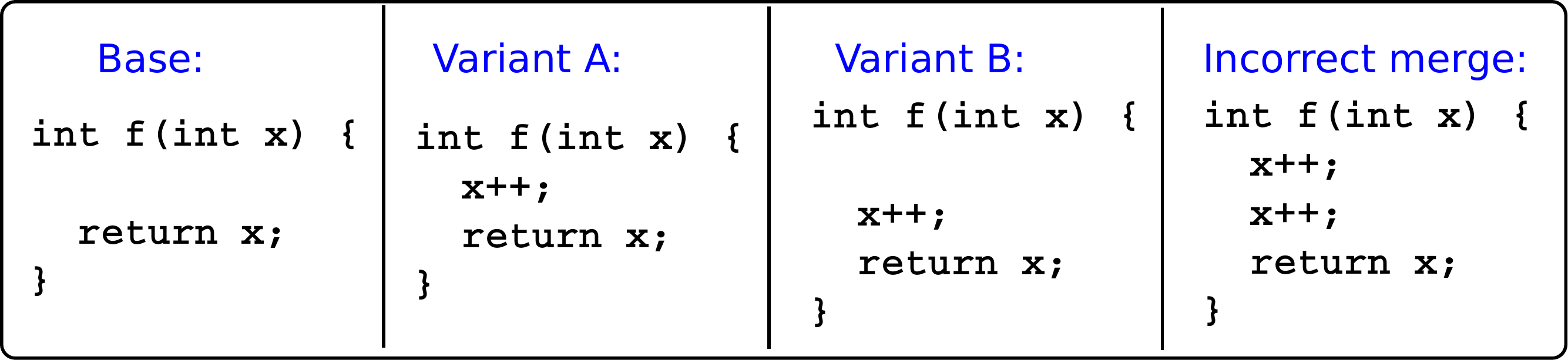}
\end{center}
\vspace{-0.15in}
\caption{Simple motivating example}\label{fig:ex-basic}
\vspace{-0.16in}
\end{figure}

\begin{figure*}[!t]
\vspace{-0.1in}
\begin{center}
\includegraphics[scale=0.3]{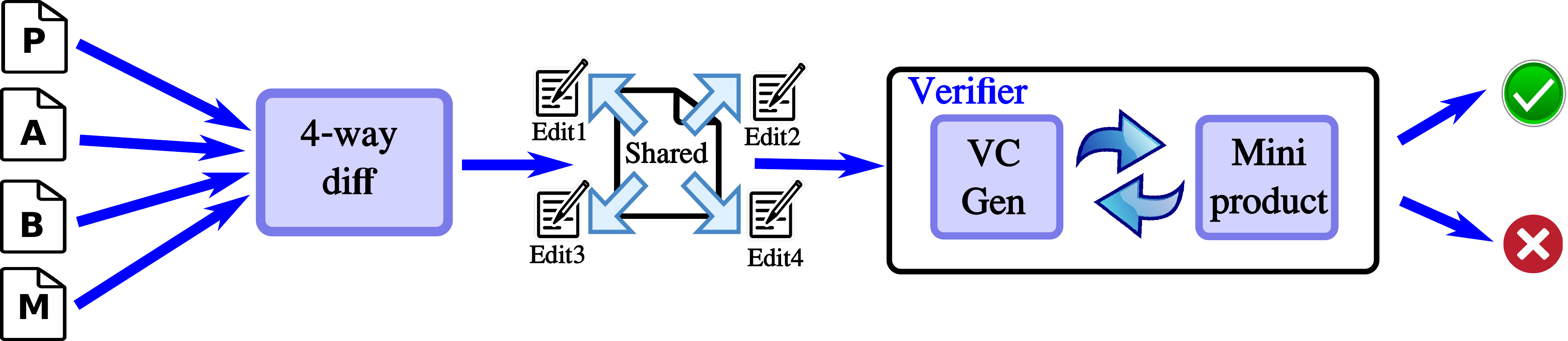}
\end{center}
\vspace{-0.35in}
\caption{High-level overview of our approach}\label{fig:workflow}
\vspace{-0.1in}
\end{figure*}

The main contribution of this paper is a novel {compositional} verification algorithm, and its implementation in a tool called \tool,  for automatically proving semantic conflict-freedom.  Our method is compositional in that it analyzes  different modifications to the program in isolation and composes them to obtain an overall proof of semantic conflict-freedom. A key idea that allows compositionality is to model different versions of the program using  \emph{edits} applied to a \emph{shared program with holes}. Specifically, the shared program captures common statements between the program versions, and  holes represent discrepancies between them. The edits describe how to fill each hole in the shared program to obtain the corresponding statement in a variant. Given such a representation that is automatically generated by \tool, our verification algorithm uses lightweight analysis to reason about shared program fragments but resorts  to precise relational techniques to reason about modifications.

The overall workflow of our approach is illustrated schematically in Figure~\ref{fig:workflow}. Our method takes as input four related programs, namely the original program $P$, two variants $A$ and $B$, and a merge candidate $M$, and represents them as edits applied to a shared program by running a ``4-way diff" algorithm on the abstract syntax trees. The verifier leverages the result of the 4-way diff algorithm to identify which parts of the program to analyze more precisely. Specifically, our verification algorithm summarizes shared program fragments using uninterpreted functions of the form $x = f(x_1, \ldots, x_n)$ that encode dependencies between program variables. In contrast, the verifier reasons about edited program fragments in a more fine-grained way by constructing 4-way  \emph{product programs} that encode the simultaneous behavior of all four edits. Overall, this interplay between lightweight dependence analysis and product construction allows our technique to generate verification conditions whose complexity depends on the size and number of the edits.

To evaluate our technique, we collect over 50 real-world merge scenarios obtained by crawling Github commit histories and evaluate \toolname\ on these benchmarks. Our tool is able to verify the correctness of the merge candidate in 75\% of the benchmarks and identifies eleven real violations of semantic conflict-freedom, some of which are not detected by textual merge tools.  Our evaluation also demonstrates the scalability of our method and illustrates the advantages of performing compositional reasoning.

In all, this paper makes the following key contributions:

\begin{itemize}[leftmargin=*]
\item We introduce the merge verification problem based on the notion of semantic conflict-freedom.
\item We provide a  compositional verification algorithm that
combines precise relational reasoning about the edits with lightweight reasoning for unedited program fragments. 
\item We present a novel $n$-way product construction technique for precise relational verification.
\item We describe an $n$-way AST diff algorithm and use it to represent program versions as edits applied to a shared program with holes.
\item We implement our method in a tool called \toolname \ and evaluate our approach on real-world merge scenarios collected from Github repositories. 
\end{itemize}

\section{Overview}
\label{sec:overview}

\renewcommand{\ttdefault}{pxtt}

\lstdefinelanguage{boogie}{
  keywords={%
    proc,free, function,returns,var,int,bool,call,return,assert,assume,goto,havoc,modifies,requires,ensures,while,if,const,axiom,foreach
  },
  morecomment=[l]{//},
  basicstyle=\ttfamily\footnotesize,
  escapechar=\%
}

\newcommand{\lstc}[1]{\lstset{language=C,basicstyle=\ttfamily\footnotesize}\lstinline #1}
\newcommand{\lstboogie}[1]{\lstset{language=boogie}\lstinline #1}

\definecolor{impacted}{rgb}{0.7,0.7,0.7}
\definecolor{syntchanged}{rgb}{1.,0.75,0.75}
\definecolor{linenumbergray}{rgb}{0.5,0.5,0.5}
\definecolor{gitdel}{RGB}{255,236,236}
\definecolor{gitdelfocused}{RGB}{248,203,203}
\definecolor{gitadd}{RGB}{234,255,234}
\definecolor{gitaddfocused}{RGB}{166,235,166}

\makeatletter
\newenvironment{btHighlight}[1][]
{\begingroup\tikzset{bt@Highlight@par/.style={#1}}\begin{lrbox}{\@tempboxa}}
{\end{lrbox}\bt@HL@box[bt@Highlight@par]{\@tempboxa}{impacted}\endgroup}

\newcommand\btHL[1][]{%
  \begin{btHighlight}[#1]\bgroup\aftergroup\bt@HL@endenv%
}
\def\bt@HL@endenv{%
  \end{btHighlight}%
  \egroup
}
\newcommand{\bt@HL@box}[3][]{%
  \tikz[#1]{%
    \pgfpathrectangle{\pgfpoint{1pt}{0pt}}{\pgfpoint{\wd #2}{\ht #2}}%
    \pgfusepath{use as bounding box}%
    \node[anchor=base west, fill=#3, outer sep=0pt,inner xsep=0pt, inner ysep=-0.5pt, #1]{\raisebox{1pt}{\strut}\strut\usebox{#2}};
  }%
}
\makeatother

\lstdefinestyle{Java}{basicstyle=\ttfamily\scriptsize,
        language=Java,
        numbers=left,
        numberstyle=\tiny\color{linenumbergray},
        moredelim=**[is][\btHL]{@i@}{@i@},
        moredelim=**[is][{\btHL[fill=gitdel]}]{@d@}{@d@},
        moredelim=**[is][{\btHL[fill=gitdelfocused, rounded corners=2pt]}]{@df@}{@df@},
        moredelim=**[is][{\btHL[fill=gitadd]}]{@a@}{@a@},
        moredelim=**[is][{\btHL[fill=gitaddfocused, rounded corners=2pt]}]{@af@}{@af@},
        escapeinside={(*@}{@*)}}

\newcommand{\letr}{\rm let}
\newcommand{\inr}{\rm in}
\newcommand{\skipr}{\rm skip}
\newcommand{\hedit}{\hat{\edit}}
\newcommand{\numHoles}[1]{\mathsf{numHoles}(#1)}
\newcommand{\tail}[1]{\mathsf{tail}(#1)}
\newcommand{\head}[1]{\mathsf{head}(#1)}
\newcommand{\Compose}{\mathsf{Compose}}
\newcommand{\GenEdit}{\mathsf{GenEdit}}
\newcommand{\DiffTwo}{\mathsf{Diff2}}

\newcommand{\sel}{\mathit select}
\newcommand{\update}{\mathit update}


In this section, we give an overview of our approach with the aid of a merge example from the \mytt{RxJava} project \footnote{\url{https://github.com/ReactiveX/RxJava/commit/1c47b0c}.}.
Figure~\ref{fig:running-example} shows the Base version ($\orig$) of the  \mytt{triggerActions} method from the \mytt{TestScheduler.java} file.
The two variants $\vara$, $\varb$ and the merge $\cand$ perform the following modifications: 
\begin{itemize}[leftmargin=*]
\item Variant $\vara$ {\it moves} the statement \mytt{time = targetTimeInNanos} at line~\ref{line:hole1} to immediately after the while loop.
This modification impacts the value of the variable \mytt{time} in $\vara$ with respect to the Base version. 
\item Variant $\varb$ guards the call \mytt{current.action.call(...)} at line~\ref{line:hole2} with a condition \mytt{if(!current.isCancelled.get()) \{...\}}.
The call (at line~\ref{line:hole2}) has a side effect on the variable called \mytt{value} (we omit the implementation of this procedure). This  modification changes the effect on \mytt{value} with respect to the Base version. 
\item The merge $\cand$ incorporates both of these changes.
\end{itemize}

\begin{figure}[t]
\begin{lstlisting}[style=Java]
int time;  int value;
void triggerActions(long targetTimeInNanos) {
    while(!queue.isEmpty()){
         TimedAction current = queue.peek();
         if(current.time > targetTimeInNanos){
             time = targetTimeInNanos; (*@\label{line:hole1}@*)
             break; 
         }
         time = current.time;
         queue.remove();
         current.action.call(current.scheduler, current.state);  (*@\label{line:hole2}@*) 
    }  } (*@\label{line:hole3}@*)
\end{lstlisting}
\vspace{-0.2in}
\caption{Procedure from  the base program in {\tt RxJava}. }
\label{fig:running-example}
\vspace{-0.2in}
\end{figure}

This example is interesting in that both variants modify code within a loop, and one of them (namely, $\varb$) changes the control-flow by introducing a conditional. 
The loop in turn depends on the state of an unbounded collection \mytt{queue}, which is manipulated using methods such as \mytt{queue.isEmpty} and \mytt{queue.remove}.
Furthermore, while \mytt{triggerActions} has no return value, it has implicit side-effects on variables \mytt{time} and \mytt{value}, and on the collection \mytt{queue}. 
Together, these features make it challenging to ensure that the merge $\cand$ preserves  changes from both variants and does not introduce any new behavior. 


To verify semantic conflict-freedom,  our techinque represents the changes formally using a list of {\it edits} over a shared program with {\it holes}.
Figure~\ref{fig:edits-running-example} shows the shared program $\hstmt$ along with the corresponding edits $\edit_\orig, \edit_\vara, \edit_\varb, \edit_\cand$.
A hole (denoted as \mytt{<?HOLE?>}) in $\hstmt$ is a placeholder for a statement.
The shared program captures the statements that are common to all the four versions ($\orig$, $\vara$, $\varb$ and $\cand$), and the holes in $\hstmt$ represent program fragments that differ between the program versions.
An edit $\edit_{\mathcal{P}}$ for program version $\mathcal{P}$ represents a list of statements that will be substituted into the holes of the shared program to obtain $\mathcal{P}$. 

\begin{figure}[h!]
\underline{Shared program with holes ($\hstmt$)}
\begin{lstlisting}[style=Java,numbers=none]
void triggerActions(long targetTimeInNanos) {
    while (!queue.isEmpty()) {
        TimedAction current = queue.peek();
        if (current.time > targetTimeInNanos) {
           <?HOLE?>;
           break; }
        time = current.time; queue.remove();
        <?HOLE?>; 
    }
    <?HOLE?>;
}
\end{lstlisting}
\vspace{-0.2in}
\underline{Edit $\orig$ ($\edit_\orig$)}
\begin{lstlisting}[style=Java,numbers=none,xleftmargin=0.0cm]
[ time = targetTimeInNanos, current.action.call(...), skip ]
\end{lstlisting}

\underline{Edit $\vara$ ($\edit_\vara$)}
\begin{lstlisting}[style=Java,numbers=none,xleftmargin=0.0cm]
[ skip, current.action.call(...), time = targetTimeInNanos ]
\end{lstlisting}
\underline{Edit $\varb$ ($\edit_\varb$)}
\begin{lstlisting}[style=Java,numbers=none,xleftmargin=0.0cm]
[ time = targetTimeInNanos, 
  if(!current.isCancelled.get()) { current.action.call(...);}, 
  skip ]
\end{lstlisting}
\vspace{-0.1in}
\underline{Edit $\cand$ ($\edit_\cand$)}
\vspace{-0.1in}
\begin{lstlisting}[style=Java,numbers=none,xleftmargin=0.0cm]
[ skip,
  if(!current.isCancelled.get()) { current.action.call(...); },
  time = targetTimeInNanos ]
\end{lstlisting}
\vspace{-1em}
\caption{Shared program with holes and the edits.}
\vspace{-0.1in}
\label{fig:edits-running-example}
\end{figure}

Given this representation, we express {\it semantic conflict-freedom} as an assertion for each of the return variables (in this case, global variables modified by the  \mytt{triggerActions} method). 
Since the \mytt{triggerActions} method modifies \mytt{time}, \mytt{value} and \mytt{queue}, we add an assertion for each of these variables. For instance,
 we add the following assertion on the value of \mytt{time} at exit from the four versions:
\[
\small
\begin{array}{c}
(\mytt{time}_\orig = \mytt{time}_\varb  = \mytt{time}_\vara = \mytt{time}_\cand ) \bigvee  \\
\big ((\mytt{time}_\orig \neq \mytt{time}_\vara  \Rightarrow \mytt{time}_\vara = \mytt{time}_\cand ) \bigwedge \\ 
(\mytt{time}_\orig \neq \mytt{time}_\varb  \Rightarrow \mytt{time}_\varb = \mytt{time}_\cand)  \big )
\end{array}
\]
This assertion states that either (i) all  four versions have identical side-effects on \mytt{time}, or (ii) if the side-effect on $\mytt{time}_\vara$ (resp. $\mytt{time}_\varb$) differs from $\mytt{time}_\orig$, then $\mytt{time}_\cand$ in the merge should have identical side-effect as $\mytt{time}_\vara$ (resp. $\mytt{time}_\varb$). 
We add similar assertions for \mytt{value} and \mytt{queue}. 

To prove these assertions, our method assumes that all four versions start out in identical states  and then generates a \emph{relational postcondition (RPC)} $\psi$ such that the merge is semantically conflict-free if $\psi$ logically implies the added assertions. Our RPC generation engine reasons about modifications over the base program by differentiating between three kinds of statements:

\vspace{-0.1in}
\paragraph{Shared statements.} We summarize the behavior of  shared statements using  straight-line code snippets of the form $y = f(x_1, \ldots, x_n)$ where  $f$ is an uninterpreted function. Essentially, such a statement indicates that the value of variable $y$ is some (unknown) function of variables $x_1, \ldots, x_n$. These ``summaries" are generated using lightweight dependence analysis and allow our method to perform {\it abstract reasoning} over unchanged program fragments. 

\vspace{-0.1in}
\paragraph{Holes.} When our RPC generation engine encounters a hole in the shared program, it performs precise relational reasoning about different modifications by computing a \emph{4-way product program} of the edits. As is well-known in the relational verification literature~\cite{product1,product2}, a product program $P_1 \times P_2$ is semantically equivalent to $P_1; P_2$ but is constructed in a way that facilitates the verification task. However, because product construction can result in a significant blow-up in program size, our technique generates \emph{mini-products} by considering each hole in isolation rather than constructing a full-fledged product of the four program versions.

\vspace{-0.1in}
\paragraph{Loops.} Our RPC generation engine  infers \emph{relational loop invariants} for loops that contain edited program fragments. 
For instance, our method infers that (i) ${\tt time}_\orig = {\tt time}_\varb$ and ${\tt time}_\vara = {\tt time}_\cand$,  (ii) ${\tt value}_\orig = {\tt value}_\vara$ and ${\tt value}_\varb = {\tt value}_\cand$, and (iii) the state of collection ${\tt queue}$ is identical in all four versions for the shared loop from Figure~\ref{fig:edits-running-example}.

Using these ideas, our method is able to automatically generate an RPC that implies semantic conflict-freedom of this example.
Furthermore, the entire procedure is push-button, including the generation of edits, RPC computation, and relational loop invariant generation.
\section{Representation of Program Versions}
\label{sec:problem}

\begin{figure}
\[
\begin{array}{llll}
\mathrm{Program \ version} & \pv & := & (\hstmt, \edit) \\
\mathrm{Edit} & \edit & := & [ \ ] \ | \ \stmt::\edit \\
\mathrm{Stmt \ with \ hole} &  \hstmt & := & \hole \ |  \ A \ | \ \hstmt_1; \hstmt_2 \ |  \ \ifstmt{C}{\hstmt_1}{\hstmt_2} \\
& & | \ & \while{C}{\hstmt} \\
\mathrm{Stmt} &  \stmt & := & A \ | \ \stmt_1; \stmt_2 \ | \ \ifstmt{C}{\stmt_1}{\stmt_2} \\ 
& & | \ & \while{C}{\stmt}  \\
{\rm Atom} & A & := & {\rm skip} \ | \ x:=e \ | \ x[e_1] := e_2  \\
\end{array}
\]
\vspace{-0.2in}
\caption{Representation of program versions. Here,  $::$ denotes list concatanation, and $e$ and $C$ represent expressions and predicates respectively. }\label{fig:lang}
\vspace{-0.15in}
\end{figure}


In this section, we describe our representation of program versions as \emph{edits} applied to a \emph{shared program with holes}. As shown in Figure~\ref{fig:lang}, a program version $\pv$ is  a pair $(\hstmt, \edit)$ where $\hstmt$ is a statement with \emph{holes} (i.e., missing statements) and an edit $\edit$ is a list of statements (without holes). Given a program version $\pv = (\hstmt, \edit)$, we can obtain a full program $\prog = \hstmt[\edit]$  by applying the edit $\edit$ to $\hstmt$ according to the $\applyEdit$ procedure of Figure~\ref{alg:apply}. Effectively, $\applyEdit$  traverses the AST in depth-first order and replaces each hole with the next statement in the edit. Given $n$ related programs $\prog_1, \ldots, \prog_n $, we assume the existence of a \emph{diff} procedure that generates a shared program $\hstmt$ as well as $n$ edits $\edit_1, \ldots, \edit_n$ such that $\forall i \in [1,n]. \ \applyEdit(\hstmt, \edit_i) = \prog_i$. Since this \emph{diff} procedure is orthogonal to our verification algorithm, we  defer the discussion of our \emph{diff} procedure until Section~\ref{sec:edit-gen}.


\begin{figure}
   \small
\[
\begin{array}{l}
\ \ \mathsf{ApplyEdit}(\hstmt, \edit) = \stmt \ \ {\rm where} (\stmt, \nil) = \mathsf{Apply}(\hstmt, \edit) \vspace{0.05in} \\ 
\ \ \apply:: (\hstmt, \edit) \rightarrow (\stmt, \Delta') \vspace{0.06in} \\
\begin{array}{lll}
\apply  (\hole, \stmt::\edit) & = &   (\stmt, \edit) \\
\apply  (A, \edit) & = & (A, \edit) \\
 \apply (\hstmt_1; \hstmt_2, \edit) & = &   {\rm let} \ (\stmt_1, \edit_1) \ = \ \apply(\hstmt_1, \edit) \ {\rm in}\\
 & &
 {\rm let} \ (\stmt_2, \edit_2) \ = \ \apply (\hstmt_2, \edit_1) \ {\rm in} \\ 
 & & ((\stmt_1; \stmt_2), \edit_2) \\
 \apply  (\ifstmt{C}{\hstmt_1}{\hstmt_2}, \edit) & = &   {\rm let} \ (\stmt_1, \edit_1) \ = \ \apply(\hstmt_1, \edit) \ {\rm in}  \\
  & &  {\rm let} \ (\stmt_2, \edit_2) \ = \ \apply(\hstmt_2, \edit_1) \ {\rm in} \\ 
 & & (\ifstmt{C}{\stmt_1}{\stmt_2}, \edit_2) \\
 \apply (\while{C}{\hstmt}, \edit) & = &   {\rm let} \ (\stmt, \edit') \ = \ \apply(\hstmt, \edit) \ {\rm in}   \\ 
 & & (\while{C}{\stmt}, \edit')
 \end{array}
\end{array}
\]
\vspace{-0.17in}
\caption{Application of edit $\Delta$ to program with holes $\hstmt$}\label{alg:apply} 
\vspace{-0.13in}
\end{figure}

Since the language from Figure~\ref{fig:lang} uses standard imperative language constructs (including arrays), we assume an operational semantics  described using judgments of the form $\sigma \vdash \stmt \Downarrow \sigma'$, where $\sigma$ is a \emph{valuation} that specifies the values of free variables in $\stmt$.
Specifically, a valuation is a mapping from (variable, index) pairs to their corresponding values.  The meaning of this judgment is that evaluating $\stmt$ under $\sigma$ yields a new valuation $\sigma'$. 
In the rest of this paper, we also assume the existence of a special array called \emph{out} that serves as the return value of the program. Any behavior  that the programmer considers relevant (e.g., side effects or writing to the console) can be captured by storing the relevant values into this \emph{out} array.


\remove{
\begin{figure}
\[
\begin{array}{ccc}
\irule{
}{\sigma \vdash {\rm skip} \Downarrow \sigma } & 
\irule{
\begin{array}{cc}
\sigma \vdash e \Downarrow c & \sigma' = \sigma[(x ,0) \mapsto c]
\end{array}
}{\sigma \vdash x := e\Downarrow \sigma' } 
&
\irule{
\begin{array}{c}
\sigma[(x,0)] = c
\end{array}
}{\sigma \vdash out(x)\Downarrow \sigma } \\ \ \\ 

\irule{
\begin{array}{c}
\sigma \vdash e_1 \Downarrow c_1  \ \ \sigma \vdash e_2 \Downarrow c_2 \\
 \sigma' = \sigma[(x ,c_1) \mapsto c_2]
\end{array}
}{\sigma \vdash x[e_1] := e_2\Downarrow \sigma'} &

\irule{
\begin{array}{c}
\sigma \vdash S_1 \Downarrow \sigma_1 \\
\sigma_1 \vdash S_2 \Downarrow \sigma_2
\end{array}
}{\sigma \vdash S_1; S_2 \Downarrow \sigma_2 }  &
\irule{
\begin{array}{c}
\sigma \vdash C \Downarrow {\rm true} \\
\sigma \vdash S_1 \Downarrow \sigma_1 \\
\end{array}
}{\sigma \vdash \ifstmt{C}{S_1}{S_2} \Downarrow \sigma_1 } \\ \ \\ 
\irule{
\begin{array}{c}
\sigma \vdash C \Downarrow {\rm false} \\
\sigma \vdash S_2 \Downarrow \sigma_2 \\
\end{array}
}{\sigma \vdash \ifstmt{C}{S_1}{S_2} \Downarrow \sigma_2 } &
\irule{
\begin{array}{c}
\sigma \vdash C \Downarrow {\rm false} 
\end{array}
}{\sigma \vdash \while{C}{S} \Downarrow \sigma } & 
\irule{
\begin{array}{c}
\sigma \vdash C \Downarrow {\rm true} \\
\sigma \vdash S \Downarrow \sigma_1 \\
\sigma_1 \vdash  \while{C}{S} \Downarrow \sigma_2
\end{array}
}{\sigma \vdash \while{C}{S} \Downarrow \sigma_2 } 
\end{array}
\]
\caption{Operational semantics. Here, $\sigma$ maps (variable, index) pairs to values, and we view scalar variables as arrays with a single valid index at 0. Since the semantics of expressions is completely standard, we do not show them here.  However, we assume that reads from locations that have not been initialized yield a special constant $\bot$.}\label{fig:semantics}
\end{figure}
}

\section{Semantic Conflict Freedom} \label{sec:conflict}

In this section, we first introduce  \emph{syntactic} conflict-freedom, which corresponds to the criterion used by many existing merge tools.  We then explain why it falls short and  formally describe the more robust notion of \emph{semantic} conflict-freedom.

\begin{definition}{(\bf Syntactic conflict freedom)}\label{def:syntactic}
Suppose that we are given four program versions $\orig = (\hstmt, \edit_\orig)$, $\vara=(\hstmt, \edit_\vara)$, $\varb=(\hstmt, \edit_\varb)$, $\cand=(\hstmt, \edit_\cand)$ representing the base program, the two variants, and the merge candidate respectively. We say that the merge candidate $\cand$ is \emph{syntactically conflict free} if the following conditions are satisfied for all $i \in [0,n)$, where $n$ denotes the number of holes in $\hstmt$:
\begin{enumerate}
\item If $\edit_\orig[i] \neq \edit_\vara[i]$, then $\edit_\cand[i] = \edit_\vara[i]$
\item If $\edit_\orig[i] \neq \edit_\varb[i]$, then $\edit_\cand[i] = \edit_\varb[i]$
\item Otherwise,   $\edit_\orig[i] = \edit_\vara[i] = \edit_\varb[i] = \edit_\cand[i]$
\end{enumerate}
\end{definition}

Intuitively, the above definition  states that the candidate merge $\cand$ makes the same syntactic change as variant $\vara$ (resp. $\varb$)  whenever $\vara$ (resp. $\varb$) differs from $\orig$. While this definition may seem intuitively sensible, it does not accurately capture what it means for a merge candidate to be correct. In particular, some incorrect merges may be conflict-free according to the above definition, while some correct merges may be rejected.

\begin{example}\label{ex:fn}
Consider  $\hstmt = \hole; \hole; out[0] := x$ and the edits $\edit_\orig = [\skipp, \skipp]$, $\edit_\vara = [x:=x+1, \skipp]$, $\edit_\varb = [\skipp, x:=x+1]$, and $\cand = [x:=x+1; x:=x+1]$.  Observe that applying these edits to $\hstmt$ yields the same programs given in Figure~\ref{fig:ex-basic}. These programs are conflict-free according to the syntactic criterion given in Definition~\ref{def:syntactic}, but the merge is clearly incorrect (both variants increment $x$ by $1$, but the merge candidate ends up incrementing $x$ by $2$). 
\end{example}

The above example illustrates that a syntactic notion of conflict freedom is not suitable for ruling out incorrect merges. Similarly, Definition~\ref{def:syntactic} can also result in the  rejection of perfectly valid merge candidates.

\begin{example}\label{ex:fp}
Consider the base program $\ifstmt{x>0}{y:=1}{y:=0}; \ out[0] := y$. Suppose this program has a bug that is caused by using the wrong predicate, so one variant fixes the bug by swapping the then and else branches, and the other variant changes the predicate from $x>0$ to $x \leq 0$. Clearly, choosing either variant as the merge  would be acceptable because they are semantically equivalent. However, there is no merge candidate that can satisfy Definition~\ref{def:syntactic} because the shared program is $\hole; out[0]:=y$ and the two variants fill the hole in syntactically conflicting ways.
\end{example}

Based on the shortcomings of syntactic conflict freedom, we instead propose the following \emph{semantic} variant:

\begin{definition}{(\bf Semantic conflict freedom)}\label{def:semantic}
Suppose that we are given four program versions $\orig, \vara, \varb, \cand $ representing the base program, its two variants, and the merge candidate respectively. We say that $\cand$ is \emph{semantically conflict-free}, if for all valuations $\sigma$ such that:
\[
\begin{array}{llll}
\sigma \vdash \orig \Downarrow \sigma_\orig &
\sigma \vdash \vara \Downarrow \sigma_\vara &
\sigma \vdash \varb \Downarrow \sigma_\varb &
\sigma \vdash \cand \Downarrow \sigma_\cand
\end{array} 
\]
the following conditions hold for all $i$:~\footnote{We assume that $out[i]$ is a special value $\bot$ if $(out, i) \not \in \emph{dom}(\sigma)$}
\begin{enumerate}[leftmargin=*]
\small
\item If $\sigma_\orig[(out, i)] \neq \sigma_\vara[(out, i)]$, then $\sigma_\cand[(out, i)] = \sigma_\vara[(out, i)]$
\item If $\sigma_\orig[(out, i)] \neq \sigma_\varb[(out, i)]$, then $\sigma_\cand[(out, i)] = \sigma_\varb[(out, i)]$
\item Otherwise,   $\sigma_\orig[(out, i)] = \sigma_\vara[(out, i)] = \sigma_\varb[(out, i)] = \sigma_\cand[(out, i)]$
\end{enumerate}
\end{definition}

In contrast to  syntactic conflict freedom, Definition~\ref{def:semantic} requires agreement between the \emph{values} that are returned by the program. Specifically, it says that, if the $i$'th value returned by variant $A$ (resp. $B$) differs from the $i$'th value returned by base, then the $i$'th return value of the merge  should agree with $A$ (resp. $B$). According to this definition, the merge candidate from Example~\ref{ex:fn} is \emph{not} conflict-free because it returns $2$ whereas both variants return  $1$. Furthermore, for Example~\ref{ex:fp}, we can find a merge candidate (e.g., one of the variants) that satisfies  semantic conflict freedom. 



\newcommand{\spost}{post}
\newcommand{\cond}{\varphi}
\newcommand{\scripts}{\vec{\edit}}
\newcommand{\inv}{\mathcal{I}}
\newcommand{\crp}{\circledast}
\newcommand{\gen}{\rightsquigarrow}
\newcommand{\lp}[2]{\while{#1}{#2}}
\newcommand{\preds}{\mathsf{Preds}}

\section{Verifying Semantic Conflict Freedom}
\label{sec:verify}

\begin{algorithm}[t]
\caption{Algorithm for verifying  conflict freedom}\label{alg:verify}
\begin{algorithmic}[1]  
\vspace{0.05in}
\Procedure{Verify}{$\hstmt, \edit_1, \edit_2, \edit_3, \edit_4$} 
\vspace{0.05in}
\State {\bf assume} $\emph{vars}(\set{\hstmt[\edit_1],\ldots, \hstmt[\edit_4]}) = V$ 
\State $\cond$ := $ (V_1 = V_2 \land V_1 = V_3 \land V_1 = V_4)$ 
\vspace{0.02in}
\State $\psi$ := $\mathsf{RelationalPost}(\hstmt, \edit_1, \edit_2, \edit_3, \edit_4, \cond)$
\vspace{0.02in}
\State $\chi_1$ := $\forall i. \ (out_1[i] \neq out_2[i] \Rightarrow out_2[i] = out_4[i]) $ 
\State $\chi_2$ := $\forall i.  \ (out_1[i] \neq out_3[i] \Rightarrow out_3[i] = out_4[i]) $
\State $\chi_2$ := $\forall i.  \ (out_1[i] = out_2[i] = out_3[i] = out_4[i])$
\vspace{0.02in}
 \State \Return {$\psi \models (\chi_1 \land \chi_2) \lor \chi_3$} 
\EndProcedure
\end{algorithmic}
\end{algorithm}

We now turn our attention to the verification algorithm for proving  semantic conflict-freedom. The high-level structure of the verification algorithm is quite simple and is shown in Algorithm~\ref{alg:verify}. It takes as input a shared program (with holes) $\hstmt$, an edit  $\edit_1$ for the base program, edits $\edit_2, \edit_3$ for the variants, and an edit  $\edit_4$ for the merge candidate. Conceptually, the algorithm consists of three steps:

\vspace{-0.05in}
\paragraph{Precondition.} Algorithm~\ref{alg:verify} starts by generating a pre-condition $\cond$ (line 3) stating  that {all} variables initially have the same value. ~\footnote{Observe that this precondition also applies to  local variables, not just arguments, and allows our technique to handle cases in which one of the variants introduces a new variable. } Note that $V_1$ denotes the variables in the base program, $V_2, V_3$ denote variables in the variants, and $V_4$ refers to variables in the merge candidate. We use the notation $V_i = V_j$ as short-hand for $\forall v \in V. \ v_i = v_j$. 

\vspace{-0.05in}
\paragraph{RPC computation.} The next step of the algorithm is to compute a \emph{relational post-condition} $\psi$ of $\cond$ with respect to the four program versions (line 4). Such a relational post-condition $\psi$ states relationships between variables $V_1, V_2, V_3$, and $V_4$ and has the property that it is also post-condition of the program $ (\hstmt[\edit_1])[V_1/V]; \ldots; (\hstmt[\edit_4])[V_4/V]$. We will explain the $\mathsf{RelationalPost}$ procedure in detail shortly.

\vspace{-0.05in}
\paragraph{Checking conflict freedom.}  The last step of the algorithm checks whether the relational post-condition $\psi$ logically implies semantic conflict freedom (line 8). Specifically,  observe that the constraint $(\chi_1 \land \chi_2) \lor \chi_3$ encodes precisely the three conditions from Definition~\ref{def:semantic}, so the program is conflict-free if $\psi$  implies $(\chi_1 \land \chi_2) \lor \chi_3$.

\subsection{Computing Relational Postconditions}\label{sec:relpost}

Since the core part of the verification algorithm is the computation of RPCs, we now describe the $\mathsf{RelationalPost}$ procedure. As mentioned in Section~\ref{sec:intro}, the key idea is to analyze edits in a precise way by constructing   product programs, but perform lightweight reasoning for shared program parts using dependence analysis.

Our RPC generation engine is described in Figure~\ref{fig:relpost} using judgments  $
\vec{\edit}, \cond \vdash \hstmt: \cond', \vec{\edit'}
$.
Here, $\cond$ is a precondition relating variables in different program versions, and $\vec{\edit}$ is a vector of $n$ edits applied to a shared base program $\hstmt$. The meaning of this judgment is that the following Hoare triple is valid:
\[ \{\cond \} \ \hstmt[\edit_1][V_1/V]; \ldots ;\hstmt[\edit_n][V_n/V] \ \{\cond'\}\] 
In other words, $\varphi'$ is a sound relational post-condition of the four program versions with respect to precondition $\varphi$. Since the edits in $\vec{\edit}$ may contain more statements than there are holes in $\hstmt$,  we use $\vec{\edit'}$ to denote the remaining edits  that were not ``used'' while analyzing $\hstmt$.

\begin{figure}
\small
\[
\begin{array}{lc}
(1) \ \ & \irule{
\begin{array}{c}
 \stmt = \textsf{head}(\edit_1)[V_1/V] \crp \ldots \crp \textsf{head}(\edit_4)[V_4/V] 
\end{array}
}{\scripts, \cond \vdash  \hole: \spost(\stmt, \cond),  [\textsf{tail}(\edit_1), \ldots, \textsf{tail}(\edit_4) ] } \\ \ \\ 
(2) \ \ & \irule{
\begin{array}{c}
\textsf{Modifies}(\stmt) = \{ {y_1}, \ldots, {y_n} \} \\
\vec{{x}_i} = \textsf{Dependencies}(\stmt, {y}_i) \\
{\stmt}_i = ({y}_i := F_i(\vec{{x}_i}))[V_1/V]; \ldots; ({y}_i := F_i(\vec{{x}_i}))[V_4/V] 
\end{array}
}{\scripts, \cond \vdash   \stmt: \spost({\stmt}_1; \ldots; {\stmt}_n, \cond), \scripts}  \\ \ \\ 
(3) \ \ \  &
\irule{
\begin{array}{ll}
\scripts, \cond   \vdash   \hstmt_1: \cond', \scripts'  &
\scripts', \cond'    \vdash   \hstmt_2: \cond'', \scripts''
\end{array}
}{
\scripts, \cond \vdash  \ \hstmt_1; \hstmt_2:  \cond'', \scripts''
} \\ \ \\ 
(4) \ \ \  &
\irule{
\begin{array}{c}
\cond \models \bigwedge_{i,j} C[V_i/V] \leftrightarrow  C[V_j/V]  \\
\scripts, \cond \land C[V_1/V]  \vdash   \hstmt_1: \cond', \scripts'  \\
\scripts', \cond \land \neg C[V_1/V]  \vdash   \hstmt_2: \cond'', \scripts'' \\

\end{array}
}{
\scripts, \cond \vdash  \ifstmt{C}{\hstmt_1}{\hstmt_2}: \cond' \lor \cond'', \scripts''
} \\ \ \\ 
 (5) \ \ & \irule{
\begin{array}{c}
\cond \models \inv \ \ \ \ \ \scripts,\inv \land \bigwedge_i  C[V_i/V] \vdash \hstmt: \inv', \scripts' \ \ \ \ \ 
\inv' \models \inv \\
\inv \models \bigwedge_{i,j} C[V_i/V] \leftrightarrow  C[V_j/V] 
\end{array}
}
{\scripts, \cond \vdash {\rm while}(C) \  \ \hstmt: \inv \land \bigwedge_i \neg C[V_i/V], \scripts' } \\ \ \\
(6) \ \ & \irule{
\begin{array}{c}
\stmt = (\hstmt[\edit_1])[V_1/V] \crp \ldots \crp (\hstmt[\edit_4])[V_4/V] \\
\edit_i = (\edit_{i}^1::\edit_{i}^2) \ \ (|\edit_{i}^1| = \mathsf{numHoles}(\hstmt))
\end{array}
}
{\scripts, \cond \vdash \hstmt: post(\stmt, \varphi), [\edit_1^2, \ldots, \edit_4^2]}
\end{array}
\]
\vspace{-0.15in}
\caption{RPC inference}
\vspace{-0.1in}
\label{fig:relpost}
\end{figure}

Let us now consider the rules in Figure~\ref{fig:relpost} in more detail. The first rule corresponds to the case where we encounter a hole in the shared program and need to analyze the edits. In this case, we construct a ``mini'' product program $\stmt$ that describes the simultaneous execution of the edits. As we will see in Section~\ref{sec:product},  an $n$-way product program $\stmt_1 \crp \ldots \crp \stmt_n$ is semantically equivalent to the sequential composition $\stmt_1; \ldots; \stmt_n$ but has the advantage of being easier to analyze. Given such a ``mini product'' $\stmt$, our RPC generation engine computes the post-condition of $\stmt$ in the standard way using a \emph{post} function, where $post(\stmt, \cond)$ yields a sound post-condition of $\cond$ with respect to $\stmt$.  Since $\stmt$ may contain loops in the general case,  the computation of \emph{post} may require loop invariant generation. As we discuss in Section~\ref{sec:product}, the key advantage of constructing a product program is to facilitate loop invariant generation using standard techniques. 

Rule (2) corresponds to the case where we encounter a program fragment $\stmt$ without  holes. Since $\stmt$ has not been modified by any of the variants, we analyze $\stmt$ in a lightweight way using dependence analysis. Specifically, for each variable $y_i$ that is modified by $\stmt$, we compute the set of variables $x_1, \ldots, x_k$ that it depends on. We then ``summarize'' the behavior of $\stmt$ using statements of the form $y_i = F_i(x_1, \ldots, x_k)$ where $F_i$ is a fresh uninterpreted function symbol. Hence, rather than analyzing the entire code fragment $\stmt$ (which could potentially be very large), we analyze its behavior in a lightweight way by modeling it as straight-line code over uninterpreted functions.~\footnote{There are rare cases in which this abstraction would lead to imprecision.  Section~\ref{sec:impl} describes how our implementation handles such cases.}

Rule (3) for sequencing is  similar to its corresponding proof rule in standard Hoare logic: Given a statement $\hstmt_1; \hstmt_2$, we first compute the relational post-condition $\cond'$ of $\hstmt_1$ and then use $\cond'$ as the precondition for $\hstmt_2$. Since $\hstmt_1$ and $\hstmt_2$ may contain edits nested inside them, this proof rule combines reasoning about $\hstmt_1$ and $\hstmt_2$ in a precise, yet lightweight way, without constructing a 4-way product for the entire program. 

Rule (4)  allows us to analyze conditionals $\ifstmt{C}{\hstmt_1}{\hstmt_2}$ in a modular way whenever possible. As in the sequencing case, we would like to analyze $\hstmt_1$ and $\hstmt_2$ in isolation and then combine the results. Unfortunately, such compositional reasoning is only possible if all program versions take the same  path.
For instance, consider the shared program $\hole; \ifstmt{x>0}{y:=1}{y:=2}$ and two versions $A, B$ given by the edits $[x:=y]$ and $[x:=z]$. 
Since $A$ could take the then branch while $B$ takes the else branch (or vice versa), we need to reason about all possible combinations of paths. 
Hence, the first premise of this rule  checks whether each $C[V_i/V]$ can be proven to be equivalent to all other $C[V_j/V]$'s under precondition $\cond$. If this is the case, all program versions  take the same path, so we can reason compositionally.  Otherwise,  our analysis falls back upon the conservative, but non-modular, proof rule (6) that we will explain shortly.

Rule (5) uses \emph{inductive  relational  invariants} for loops that have been edited in different ways by each program variant. Specifically, the first premise of this rule states that the relational invariant $\inv$ is implied by the loop pre-condition, and the next two premises enforce that $\inv$ is preserved by the loop body (i.e., $\inv$ is inductive). Thus, assuming that all loops execute the same  number of times (checked by line 2 of rule 5), we can conclude that $\inv \land \bigwedge_i \neg C[V_i/V]$ holds after the loop.
Note that rule (5) does not describe how to compute such relational loop invariants; it simply asserts that $\inv$ is inductive. As we describe in Section~\ref{sec:impl}, our implementation uses standard techniques based on conjunctive predicate abstraction to infer such relational loop invariants. 

Rule (6) allows us to fall back upon non-modular reasoning when it is not sound to analyze edits in a compositional way. Given a statement $\hstmt$ with holes, rule (6) constructs the product program $(\hstmt[\edit_1])[V_1/V] \crp \ldots \crp (\hstmt[\edit_4])[V_4/V] $ and computes its post-condition in the standard way. While rule (6) is  a generalization of rule (1),  it is only used in cases where compositional reasoning is unsound, as 
 product construction can cause a blow up in program size.


\begin{theorem}{\bf (Soundness of relational post-condition) }~\footnote{
\ifsupp
Proofs of all theorems are available in the Appendix.
\else {Proofs of all theorems are available under supplementary materials.}
\fi} Let $\hstmt$ be a shared program with holes and $\vec{\edit}$ be the edits such that $|\edit_i|  = \mathsf{numHoles}(\hstmt)$. Let $\cond'$ be the result of  calling $\mathsf{RelationalPost}(\hstmt, \vec{\edit}, \cond)$ (i.e., $\vec{\edit}, \cond \vdash \hstmt: \cond', [ ]$ according to Figure~\ref{fig:relpost}). Then, the following Hoare triple is valid:
\[ \{ \cond \} \  (\hstmt[\edit_1])[V_1/V]; \ldots; (\hstmt[\edit_n])[V_n/V] \ \{ \cond' \} \]
\end{theorem}


\subsection{Construction of Product Programs}
\label{sec:product}
In this section, we describe our method for constructing~$n$-way 
product programs. 
While there are several strategies for generating~$2$-way product 
programs in the literature (e.g.,~\cite{product1,product2}), our 
method differs from these approaches in that 
it uses similarity metrics to guide product construction and also generalizes 
these techniques to $n$-way products. The use of similarity metrics allows our method
to generate more verification-friendly product programs while obviating the need for
performing backtracking search over non-deterministic product construction rules.
%


Before we describe our product construction technique, 
we first give a simple example to illustrate how product 
construction facilitates relational verification:
\begin{example}\label{ex:motivate-product}
Consider the following programs $\stmt_1$ and $\stmt_2$:
\[
\small
\begin{array}{ll}
\stmt_1: & i_1 := 0; \ \lp{i_1<n_1}{i_1 := i_1 * x_1} \\
\stmt_2: & i_2 := 0; \ \lp{i_2<n_2}{i_2 := i_2 * x_2} \\
\end{array}
\]
and the precondition $n_1 = n_2 \land x_1 = x_2$. 
It is easy to see that $i_1$ and $i_2$ will have the 
same value after executing $\stmt_1$ and $\stmt_2$. 
Now, consider analyzing the program $\stmt_1; \stmt_2$.
%
%
While a static analyzer can \emph{in principle} infer this 
post-condition by coming up with a precise loop invariant 
that captures the exact symbolic value of $i_1$ and $i_2$ 
during each iteration, this is clearly a very difficult task. 
To see why product programs are useful, 
now consider the following program $\stmt$:
\[
\small
\begin{array}{ll}
(1) & i_1 := 0; i_2 := 0; \\
(2) & \lp{i_1<n_1 \land i_2 < n_2}{ i_1 := i_1 * x_1;  \ i_2 := i_2 * x_2;} \\
(3) & (i_1 < n_1)? \{\lp{i_1 < n_1}{i_1 := i_1 * x_1}\} \\
    & \ \ \ \ \ \ \ \ \ \ \ \ \ \ : \{ (i_2 < n_2) ? \{\lp{i_2 < n_2}{i_2 := i_2 * x_2}\}:\{\rm skip\} \} 
\end{array}
\]
Here, $\stmt$ is  equivalent to $\stmt_1; \stmt_2$ 
because it executes both loops in lockstep until one of them terminates 
and then executes the remainder of the other loop. 
While this code may look complicated, it is much easier to 
statically reason about $\stmt$ than  $\stmt_1; \stmt_2$.
In particular, since $i_1 = i_2 \land x_1 = x_2 \land n_1 = n_2$ is an 
inductive invariant of the first loop in $\stmt$, we can easily prove
that line (3) is dead code and that $i_1 = i_2$ is a 
valid post-condition of $\stmt$. 
As this example illustrates, product programs can
make relational verification easier by executing loops from different
programs in lockstep. 
\end{example} 

Our $n$-way product construction method is presented in Figure~\ref{fig:product}
using inference rules that derive judgments of the form 
$\vdash \stmt_1 \crp \ldots \crp \stmt_n \gen \stmt$
where programs $\stmt_1, \ldots, \stmt_n$ \emph{do not share any variables}
(i.e., each  $\stmt_i$ refers to variables $V_i$ such that $V_i \cap V_j = \emptyset$ for $i \neq j$).
The generated product   $\stmt$ is semantically equivalent to 
$\stmt_1; \ldots; \stmt_n$ but is constructed in  a way that makes $\stmt$
easier to be statically analyzed.
Similar to prior relational verification techniques, 
the key idea is to synchronize loops from different program 
versions as much as possible. However, our method differs from existing techniques in that it 
uses similarity metrics to 
guide product construction and generalizes them to $n$-way products.

\vspace{0.05in} \noindent
{\bf \emph{Notation.}} 
Before discussing Figure~\ref{fig:product}, we first introduce 
some useful notation: We abbreviate $\stmt_1 \crp \ldots \crp \stmt_n$ using the notation 
$\mfrak{P}^{\crp}$, and we write $\mfrak{P}$ to denote the list $(\stmt_1, \ldots, \stmt_n)$.
Also, given a statement $\stmt$, we write $\stmt[i]$ to denote the $i$'th element in the sequence (i.e., $\stmt[0]$
denotes the first element).

\vspace{0.05in} \noindent
{\bf \emph{Similarity metric.}}  As mentioned earlier, our 
algorithm uses similarity metrics between different program fragments to
guide product construction. Thus, our algorithm is parameterized by a function
$\similar \colon \stmt^* \to \real^{+}_{0}$
that returns a positive real number
representing similarity between different statements. While the precise definition of 
$\similar$ is orthogonal to our product construction algorithm, our implementation uses Levensthein distance as
the similarity metric.
%
%
%

\begin{figure}
\small
\[\arraycolsep=0.4pt
\begin{array}{lc}
(1) \ \ & \irule{
\begin{array}{c}
\vdash \stmt_1 \crp \mfrak{P}^{\crp} \gen \stmt
\end{array}
}{\vdash A;\stmt_1 \crp \mfrak{P}^{\crp} \gen A; \stmt } \\ \ \\
(2) \ \ & \irule{
\begin{array}{c}
\vdash \stmt_t;\stmt_1 \crp \mfrak{P}^{\crp} \gen \stmt'  \ \ \ \ \vdash \stmt_e;\stmt_1 \crp \mfrak{P}^{\crp} \gen \stmt''
\end{array}
}{\vdash (\ifstmt{C}{\stmt_{t}}{\stmt_{e}});\stmt_1 \crp \mfrak{P}^{\crp} \gen (\ifstmt{C}{\stmt'}{\stmt''}) } \\ \ \\ 
%
%
%
(3) \ \ & \irule{
\begin{array}{c}
\exists \stmt_i \in \mfrak{P}.\ \stmt_i[0] \neq \lp{C_i}{\stmt_{B_i}} \\
\vdash \stmt_i \crp \mfrak{(P \setminus \stmt_i)}^{\crp} \crp (\lp{C_1}{\stmt_{B_1}});\stmt_1 \gen \stmt
\end{array}
}{\vdash (\lp{C_1}{\stmt_{B_1}});\stmt_1 \crp \mfrak{P}^{\crp}  \gen \stmt} \\ \ \\ 
(4) \ \ & \irule{
\begin{array}{c}
\forall \stmt_i \in \mfrak{P}.\ \stmt_i[0] = \lp{C_i}{\stmt_{B_i}} \\
\exists \mfrak{H} \subseteq \mfrak{P}.\ \forall \mfrak{L} \subseteq \mfrak{P}.\ \simop{\mfrak{H}} \geq \simop{\mfrak{L}} \\
\vdash (\mfrak{H}[0])^{\crp} \gen \stmt' \ \ \ \ \vdash (\mfrak{H}[1\ldots])^{\crp} \crp (\mfrak{P} \setminus \mfrak{H})^{\crp} \gen \stmt''
\end{array}
}{\mfrak{P}^{\crp} \gen \stmt'; \stmt''} \\ \ \\ 
(5) \ \ & \irule{
\begin{array}{c}
\vdash \stmt_{B_1} \crp \stmt_{B_2} \gen \stmt \\
W := \lp{C_1 \land C_2}{\stmt}  \\
R :=  \ifstmt{C_1}{\lp{C_1}{\stmt_{B_1}}}{(\ifstmt{C_2}{\lp{C_2}{\stmt_{B_2}}}{{\rm skip}})}\\
\vdash W; R \crp \mfrak{P}^{\crp} \gen \stmt'
\end{array}
}{\vdash (\lp{C_1}{\stmt_{B_1}}) \crp  (\lp{C_2}{\stmt_{B_2}}) \crp \mfrak{P}^{\crp} \gen \stmt'}
\end{array} 
\]
\vspace{-0.13in}
\caption{Product construction. The base case is the trivial rule $\vdash \stmt \gen \stmt$, and we assume that every program ends in a \emph{skip} and that $\emph{skip} \crp \mfrak{P}^\crp$ is the same as  $\mfrak{P}^\crp$.}
\label{fig:product}
\vspace{-0.2in}
\end{figure}

\vspace{0.05in}\noindent
{\bf \emph{Product construction algorithm.}} 
We are now ready to explain the product construction rules shown in Figure~\ref{fig:product}.
Rule (1) is quite simple and deals with the case where the 
first program starts with an atomic statement $A$. 
Since we can always compute a precise post-condition for atomic statements, 
it is not necessary to ``synchronize'' $A$ with any of the statements 
from other programs.
Therefore, we first compute the product program 
$\stmt_1 \crp \mfrak{P}^{\crp}$,
\ie $\stmt_1 \crp \stmt_2 \crp \ldots \crp \stmt_n$,
and then sequentially compose it with $A$.

Rule (2) 
considers the case where the first program starts with a conditional 
$\ifstmt{C}{\stmt_{t}}{\stmt_{e}}$.
In general, $\stmt_{t}$ and $\stmt_{e}$ may contain loops; 
therefore, there may be an opportunity to synchronize any 
loops within $\stmt_{t}$ and $\stmt_{e}$ with loops from 
$\mfrak{P} = \stmt_2, \ldots, \stmt_n$. Therefore, we construct the 
product program as $\ifstmt{C}{\stmt'}{\stmt''}$ where $\stmt'$ (resp. $\stmt''$) is the product of
the then (resp. else) branch with $\mfrak{P}^\crp$.~\footnote{Observe that our handling of if statements can cause a blow-up in program size, since we essentially embed the continuation $\stmt_1$ inside the then and else branches. However, because our product construction applies to small program fragments, we have not found it to be a problem in practice.}
%
%
%

Because the main point of product construction is to generate a verification-friendly program by 
executing loops in lock-step, all of the remaining rules deal with loops. Specifically, rule (3)
considers the case where the first program starts with a 
loop but there is some program~$\stmt_i$ in ~$\mfrak{P}= (\stmt_2, \ldots, \stmt_n)$ that 
does not start with a loop. In this case, we want to ``get rid of'' program $\stmt_i$ by using rules 
(1) and (2); thus, we move $\stmt_i$ to the beginning and construct the product program $\stmt$ for 
$\stmt_i \crp (\mfrak{P}\backslash\stmt_i)^\crp \crp \while{C_1}{\stmt_{B_1}}; \stmt_1$. 

Before we continue to the other rules, we  make two important observations about rule (3). First, this rule exploits the commutativity and associativity of the $\crp$ operator~\footnote{Recall that different programs do not share variables}; however, it uses these properties in a restricted form by applying them only where they are useful. Second, after exhaustively applying rules (1), (2), and (3) on some $\mfrak{P}_0^\crp$, note that we will end up with a new $\mfrak{P}_1^\crp$ where \emph{all} programs in $\mfrak{P}_1$ are guaranteed to start with a loop.

%
%
Rule (4) considers the case where all programs start with a loop and utilizes the similarity metric~$\similar$ to identify which loops to synchronize. In particular, let $H$ be the subset of the programs in $\mfrak{P}$ that are ``most similar" according to our similarity metric. Since all programs in $H$ start with a loop, we first construct the product program $\stmt'$ of these loops. We then construct the product program $\stmt''$ for the remaining programs $\mfrak{P} \backslash H$ and the remaining parts of the programs in~$H$.

%
%

The final rule (5) defines what it means to 
``execute loops in lockstep as much as possible''. Given two programs that
start with loops $\while{C_1}{\stmt_1}$ and $\while{C_2}{\stmt_2}$, we first construct the product $\stmt_1 \crp \stmt_2$ 
and generate the synchronized loop as 
$\while{C_1 \land C_2}{\stmt_1 \crp \stmt_2}$. 
Since these loops may not execute the same number of 
times, we still need to generate the ``continuation'' $R$, which 
executes any remaining iterations of one of the loops. 
Thus, $W; R$ in rule (5) is semantically 
equivalent to $\while{C_1}{\stmt_1}; \while{C_2}{\stmt_2}$. 
Now, since there may be further synchronization opportunities 
between $W; R$ and the remaining programs $\stmt_3, \ldots, \stmt_n$,
we obtain the final product program by computing 
$W; R \crp \stmt_3 \crp \ldots \crp \stmt_n$. 

\begin{example}
Consider again the programs $\stmt_1$ and $\stmt_2$ from 
Example~\ref{ex:motivate-product}.
We can use rules (1) and (5)  from Figure~\ref{fig:product} to compute  the product program
for $\stmt_1 \crp \stmt_2$. The resulting product is exactly the 
program 
$\stmt$ shown in Example~\ref{ex:motivate-product}. 
\end{example}

Since rules (4) or (5) are both applicable when all programs start with a loop,
our product construction algorithm first applies rule (4) and then uses rule (5) when 
constructing the product for $(H[0])^\crp$ in rule (4). Thus, our  method
ensures that loops that are most similar to each other are executed in lockstep, which
in turn greatly facilitates verification.

\remove{
\begin{figure}
\small
\[\arraycolsep=0.4pt
\begin{array}{lc}
(1) \ \ & \irule{
\begin{array}{c}
\exists i, j. \ |\stmt_i| \neq |\stmt_j| \\
(\stmt_1', \ldots, \stmt_n') := \mathsf{IntroduceSkips}(\stmt_1, \ldots, \stmt_n) \\
\vdash \stmt'_1 \crp \ldots \crp \stmt'_n \gen \stmt
\end{array}
}{\vdash \stmt_1 \crp \ldots \crp \stmt_n  \gen \stmt } \\ \ \\ 
(2) \ \ & \irule{
\begin{array}{c}
\forall i,j. \ |\stmt_i| = |\stmt_j| = k > 1 \\
\vdash \stmt_1[i] \crp \ldots \crp \stmt_n[i] \gen \stmt'_i
\end{array}
}{\vdash \stmt_1 \crp \ldots \crp \stmt_n  \gen \stmt'_1; \ldots; \stmt'_k } \\ \ \\ 
(3) \ \ & \irule{
\begin{array}{c}
\vdash \stmt_2 \crp \ldots \crp \stmt_n \gen \stmt
\end{array}
}{\vdash A \crp  \stmt_2 \crp \ldots \crp \stmt_n  \gen A; \stmt } \\ \ \\ 
(4) \ \ & \irule{
\begin{array}{c}
\vdash \stmt_2 \crp \ldots \crp \stmt_n \gen \stmt \\
\vdash \stmt_{11} \crp \stmt \gen \stmt' \ \ \ \ \vdash \stmt_{12} \crp \stmt \gen \stmt''
\end{array}
}{\vdash (\ifstmt{C}{\stmt_{11}}{\stmt_{12}}) \crp  \stmt_2 \crp \ldots \crp \stmt_n  \gen 
(\ifstmt{C}{\stmt'}{\stmt''}) } \\ \ \\ 
(5) \ \ & \irule{
\begin{array}{c}
\vdash \stmt_1 \crp \stmt_2 \gen \stmt \\
W := \lp{C_1 \land C_2}{\stmt}  \\
R :=  \ifstmt{C_1}{\lp{C_1}{\stmt_1}}{(\ifstmt{C_2}{\lp{C_2}{\stmt_2}}{{\rm skip}})}\\
\vdash W; R \crp \stmt_3 \crp \ldots \crp \stmt_n \gen \stmt'
\end{array}
}{\vdash (\lp{C_1}{\stmt_1}) \crp  (\lp{C_2}{\stmt_2}) \crp \stmt_3 \crp \ldots \crp \stmt_n  \gen \stmt'} \\ \ \\ 
(6) \ \ & \irule{
\begin{array}{c}
\vdash \stmt_2 \crp  (\lp{C_1}{\stmt_1}) \crp \ldots \crp \stmt_n \gen \stmt
\end{array}
}{\vdash (\lp{C_1}{\stmt_1}) \crp  \stmt_2 \crp  \ldots \crp \stmt_n  \gen \stmt} 
\end{array} 
\]
\caption{Rules describing product construction. The base case is the trivial rule $\vdash \stmt \gen \stmt$ }
\label{fig:product}\end{figure}

The first rule in Figure~\ref{fig:product} deals with the case where some of the programs $\stmt_1, \ldots, \stmt_n$ have different lengths. In this case, our approach heuristically introduces skip statements to ensure that  all programs have the same length and then constructs the product program for the modified statements $\stmt_1', \ldots, \stmt_n'$. We describe the heuristics used in our implementation in Section~\ref{sec:impl}.

The second rule handles the case  where we have multiple sequence statements of equal length. Specifically, if each $\stmt_i$ is of the form $\stmt_{i1}; \ldots; \stmt_{ik}$, we construct a separate product program $\stmt_{1j} \crp \ldots \crp\stmt_{nj}$ for each $j \in [1, k]$ and then obtain the final product by sequentially composing them.

Next, let us consider the case where the first program  is an atomic statement $A$. Since we can always compute a precise post-condition for atomic statements, it is not necessary to ``synchronize'' $A$ with any of the statements from other programs. Therefore, we first compute the  product program $\stmt_2 \crp \ldots \crp \stmt_n$ and then sequentially compose it with $A$.

Rule (4) in Figure~\ref{fig:product} concerns the case where the first program is a conditional  $\ifstmt{C}{\stmt_{11}}{\stmt_{12}}$. In general, $\stmt_{11}$ and $\stmt_{12}$ may contain nested loops; therefore, there may be an opportunity to synchronize any nested loops within $\stmt_{11}$ and $\stmt_{12}$ with loops from $\stmt_2, \ldots, \stmt_n$.  Hence, we first compute the  product of $\stmt_2, \ldots, \stmt_n$ as $\stmt$ and then obtain the new then and else branches as $\stmt_{11} \crp \stmt$ and $\stmt_{12} \crp \stmt$ respectively. 

The next rule (5) describes how to construct a product program if the first two programs are loops of the form $\while{C_1}{\stmt_1}$ and $\while{C_2}{\stmt_2}$. In this case, we would like to execute these two loops in lockstep to the extent possible. Therefore, we first construct the product $\stmt_1 \crp \stmt_2$ and generate the synchronized loop as $\while{C_1 \land C_2}{\stmt_1 \crp \stmt_2}$. Of course, since these loops may not execute the same number of times, we still need to generate the ``continuation'' $R$, which executes any remaining iterations of one of the loops. Thus, it is easy to see that $W; R$ in rule (5) is semantically equivalent to $\while{C_1}{\stmt_1}; \while{C_2}{\stmt_2}$. Now, observe that there may be further synchronization opportunities between $W; R$ and the remaining programs $\stmt_3, \ldots, \stmt_n$ if they contain nested loops. Hence, we obtain the final product program by computing $W; R \crp \stmt_3 \crp \ldots \crp \stmt_n$. 

The final rule also deals with loops, but only applies in the case where the first program is a loop but the second program $\stmt_2$ is not. In this case, we simply swap the first and second program so that $\stmt_2$ either gets ``consumed'' or we find a synchronization opportunity between a  loop nested inside $\stmt_2$ and the first program $\while{C_1}{\stmt_1}$.
}

\begin{theorem}{\bf (Soundness of product)}
\label{thm:product-sound}
Let $\stmt_1, \ldots, \stmt_n$ be statements with disjoint
variables, and let 
$\vdash \stmt_1 \crp \ldots \crp \stmt_n \gen \stmt$ according 
to~\Cref{fig:product}. 
Then, for all valuations $\sigma$, we have 
$\sigma \vdash \stmt_1; \ldots; \stmt_n \Downarrow \sigma'$ iff
$\sigma \vdash \stmt \Downarrow \sigma'$.
\end{theorem}

\section{Edit Generation}\label{sec:edit-gen}
The verification algorithm we described in Section~\ref{sec:verify} requires all program versions to be represented as edits applied to a shared program with holes. This representation is very important because it allows  our verification algorithm to reason about modifications to different program parts in a compositional way. In this section, we describe an $n$-way AST differencing algorithm that can be used to generate the desired program representation.

\begin{algorithm}[t]
\caption{$n$-way AST differencing algorithm}\label{alg:editGen}
\begin{algorithmic}[1]

\vspace{0.05in}
\Procedure{NDiff}{$\stmt_1, \ldots, \stmt_n$} 
\vspace{0.03in}
\State $\hstmt \gets \stmt_1$; \ \ \ $\vec{\edit} \gets []; \ \ \ i \gets 2; $
\While{$i \leq n$}
\vspace{0.03in}
\State $(\hstmt, \vec{\edit}) 
\gets \GenEdit(\hstmt, \stmt_i, \vec{\edit})$ 
\EndWhile
\vspace{0.03in}
\State \Return $(\hstmt, \vec{\edit})$
\EndProcedure

\Procedure{$\GenEdit$}{$\hstmt, \stmt, \edit_1, \ldots, \edit_k$}
\vspace{0.03in}
\State ($\hstmt', \edit, \hedit$) $:=$ $\DiffTwo(\stmt, \hstmt)$\label{line:calldiff2} 
\For{$i$ in $[1, k]$} \label{line:loopStart}
\State $\edit_i' := \Compose(\hedit, \edit_i)$ \label{line:callCompose}
\EndFor 
\State \Return ($\hstmt', \edit_1', \ldots, \edit_k', \edit$)
\EndProcedure

\Procedure{$\Compose$}{$\hedit, \edit$}
\vspace{0.03in}
\If{$\hedit = [ \ ]$} \ \Return $[  \ ] $ 
\ElsIf{$\head{\hedit} 	= \hole$} 
\State \Return $\head{\edit} :: \Compose(\tail{\hedit}, \tail{\edit})$ \label{line:callCompose1}
\Else \ \Return $\head{\hedit} :: \Compose(\tail{\hedit}, \edit)$ \label{line:callCompose2}
\EndIf
\EndProcedure


\end{algorithmic}
\label{alg:editGen}
\end{algorithm}

Our $n$-way diff algorithm is presented in Algorithm~\ref{alg:editGen}. Procedure {\sc NDiff} takes as input $n$ programs $\stmt_1, \ldots, \stmt_n$ and returns a pair $(\hstmt, \vec{\edit})$ where $\hstmt$ is a shared program with holes and $\vec{\edit}$ is a list of edits such that $\hstmt[\edit_i] = \stmt_i$. The loop inside the {\sc NDiff} procedure maintains the  key invariant
$
\forall j. \ 1 \leq j < i \Rightarrow \hstmt[\edit_j] = \stmt_j
$.
Thus, upon termination, {\sc NDiff} guarantees  that $\hstmt[\edit_i] = \stmt_i$ for all $i \in [1,n]$.

The bulk of the work of the {\sc NDiff} procedure is performed by the auxiliary \textsf{GenEdit} function, which uses a 2-way AST differencing algorithm to extend the diff from $k$ to $k+1$ programs. Specifically, \textsf{GenEdit} takes as input a new program $\stmt$ as well as the  diff of the first $k$ programs, where the diff is represented as a shared program $\hstmt$ with holes as well as edits $\edit_1, \ldots, \edit_k$. The key idea underlying \textsf{GenEdit} is to use a standard 2-way AST diff algorithm to compute the diff between  $\hstmt$ and the new program $\stmt$ and then use the result to update the existing edits $\edit_1, \ldots, \edit_k$.

In more detail, the \textsf{Diff2} procedure used in  \textsf{GenEdit} yields the 2-way diff of $\hstmt$ and $\stmt$ as a triple $(\hstmt', \edit, \hat{\edit})$ such that $\hstmt'[\edit] = \stmt$ and $\hstmt'[\hat{\edit}] = \hstmt$.~\footnote{Existing 2-way AST diff algorithms can be adapted to produce diffs in this form. We provide our  \textsf{Diff2} implementation under supplementary materials.} The core insight underlying  \textsf{GenEdit}  is to use $\hat{\edit}$ to update the existing edits $\edit_1, \ldots, \edit_k$ for the first $k$ programs. Specifically, we use a procedure \textsf{Compose} to combine each existing edit $\edit_i$ with the output $\hat{\edit}$ of \textsf{2Diff}. The  \textsf{Compose} procedure is defined recursively and inspects the first element of $\hedit$ in each recursive call. If the first element is a hole, we preserve the existing edit; otherwise, we use the edit from $\hedit$.  Thus, if  $\mathsf{Compose}(\hedit, \edit_i)$ yields $\edit_i'$, we have   $\hstmt'[\edit'_i] = \hstmt[\edit_i]$. In other words, the \textsf{Compose} procedure allows us to update the diff of the first $k$ programs to generate a sound diff of  $k+1$ programs.

\begin{theorem}{\bf (Soundness of NDiff)}\label{thm:ndiff}
Let $\mathsf{NDiff}(\stmt_1, \ldots, \stmt_n)$ be $(\hstmt, \vec{\edit})$. Then we have $\hedit[\edit_i] = \stmt_i$ for all $i \in [1, n]$.
\end{theorem}
\section{Implementation}\label{sec:impl} 

We implemented the techniques proposed in this paper in a tool called \toolname\ for checking semantic conflict-freedom of Java programs. \toolname\ is written in Haskell and uses the Z3 SMT solver~\cite{z3}. In what follows, we describe relational invariant generation, our handling of various aspects of the Java language and other implementation choices.

\vspace{0.05in} \noindent
{\bf \emph{Relational invariant generation.}} The RPC computation engine from Section~\ref{sec:relpost} requires an inductive loop invariant relating  variables from the four program versions. Our implementation automatically infers relational loop invariants using the Houdini framework for (monomial) predicate abstraction~\cite{Flanagan01}. 
Specifically, we consider predicate templates of the form $x_i = x_j$ relating values of the same variable from different program versions, and compute the strongest conjunct that satisfies the conditions of rule (5) of Figure~\ref{fig:relpost}.
 
\vspace{0.05in} \noindent
{\bf \emph{Modeling the heap and collections.}} 
As standard in prior verification literature~\cite{Flanagan:2002:ESC:512529.512558}, we model each field $f$ in the program as follows: We introduce a map $f$ from object identifiers to values and model reads and writes to the map  using the $\sel$ and $\update$ functions in the theory of arrays.  
Similarly, our implementation  models collections, such as \mytt{ArrayList} and \mytt{Queue}, using arrays. 
Specifically, we use an array to represent the contents of the collection and use scalar variables to model the size of the collection as well as the current position of an iterator over the collection~\cite{dillig-popl11}. 

\vspace{0.05in} \noindent
{\bf \emph{Side effects of a method.}} Our formalization uses an {\tt out} array to model all relevant side effects of a method. 
Since real Java programs do not contain such a construct, our implementation checks semantic conflict freedom on the method's return value, the final state of the receiver object as well as any field modified in the method. 

\vspace{0.05in} \noindent
{\bf \emph{Analysis of shared statements.}} Recall that our technique
abstracts away shared program statements using uninterpreted functions (rule (2) from  Figure~\ref{fig:relpost}). However, because unconditional use of such abstraction can  result in false positives, our implementation  checks for certain conditions before applying rule (2) from Figure~\ref{fig:relpost}. Specifically, given precondition $\phi$ and variables $V$ accessed by shared statement $S$, our implementation applies rule (2) only when  $\phi$ implies semantic conflict freedom on all variables in set $V$; 
otherwise, our implementation falls back on product construction (i.e., rule (6) from Figure~\ref{fig:relpost}). While this check fails rarely in practice, it is nonetheless useful for avoiding false positives.





\vspace{-0.1in}
\subsection{Limitations}
\label{sec:limitations}

Our current prototype implementation has a few limitations:

\vspace{0.05in} \noindent
{\bf \emph{Analysis scope.}} Because \toolname\ only analyzes the class file associated with the modified procedure, it may suffer from both false positives and negatives. In particular, our analysis results are only sound  under the assumption that the external callees from other classes have \emph{not} been modified. 

\vspace{0.05in} \noindent
{\bf \emph{Changes to method signature.}}
\toolname \ 
currently does not support renamed methods or methods with parameter reordering, introduction, or deletion.
However, our tool does not place any requirements on the mapping of local variables. 
Similarly, new fields can be introduced or deleted in different variants --- we assume they are present in all four versions and that they start out in an arbitrary but equal state.  

\vspace{0.05in} \noindent
{\bf \emph{Concurrency, termination, and exceptions.}}
Neither our formalism nor our prototype implementation support sound reasoning in the presence of concurrency.
Our soundness claims also rely on the assumption that none of the variants introduce non-terminating behavior.
Finally, although exceptions can be conceptually desugared in our formalism, our  implementation does not handle exceptional control flow. 

\section{Experimental Evaluation}
\label{sec:eval}

To assess the usefulness of the proposed method, we perform a 
series of three experiments. 
In our first experiment, we use \toolname\ to verify semantic 
conflict-freedom of  merges collected from Github commit histories.
In our second experiment, we run  \toolname\ on erroneous
merge candidates generated by {\tt \small kdiff3}~\cite{kdiff3}, 
a widely-used textual merge tool. 
Finally, in our third experiment, we assess the scalability of 
our method and the importance of various design choices.
All experiments are performed on Quad-core Intel Xeon CPU 
with 2.4 GHz and 8 GB memory.

\subsection{Evaluation on Merge Candidates from Github}
\label{subs:github}

To perform our first experiment, we implemented a crawler that 
examines git merge commit histories and extracts interesting 
methods that have the potential to violate conflict freedom. 
Specifically, our crawler considers a merge scenario to be 
relevant if (a) a method is modified by \emph{both} variants in 
different ways, (b) this method involves externally visible side
effects~\footnote{Our crawler considers a method to have side-effects 
if it its return value is not {\tt void} or if it makes an assignment 
to a field.}, (c) the merge candidate is  different from either of 
the variants, and (d) the code does not involve features that are not
handled by our prototype. 

To perform this experiment, we run our crawler on nine popular Java 
applications, namely 
Elasticsearch~\cite{elasticsearch},
libGDX~\cite{libgdx}, 
iosched~\cite{iosched}, 
kotlin~\cite{kotlin}, 
MPAndroidChart~\cite{mpandroidchart},
okhttp~\cite{okhttp},
retrofit~\cite{retrofit}, 
RxJava~\cite{rxjava} and the 
Spring Boot framework~\cite{spring-boot}.
Out of 1998 merge instances where a Java source 
file is modified in both variants, 235 cases involve 
modifications to the same method where the merge differs from Base, A and B.
After filtering methods with no side-effects or containing unhandled features, 
we obtain a total of 52 benchmarks and evaluate \toolname\ on all of them.~\footnote{All benchmarks
can be found under supplementary materials.}

{\small
\begin{table*}
   \centering
   \centering
      \begin{tabular}{|c | c | c | c | c | c || c | c | c | c| c | c| c|}
      \hline
       ID & App & LOC  & Time (s) &  $\begin{array}{c}  {\rm{Result}} \\ (\scriptsize \textsc{SafeMerge}) \end{array}$ & $\begin{array}{c}  {\rm{Result}} \\ (\scriptsize \texttt{kdiff3}) \end{array}$ & ID & App & LOC & Time (s) & $\begin{array}{c}  {\rm{Result}} \\ (\scriptsize \textsc{SafeMerge}) \end{array}$ & 
       $\begin{array}{c}  {\rm{Result}} \\ (\scriptsize \texttt{kdiff3}) \end{array}$ \\
      \hline
      \hline
          1  & ESearch & 18  & 0.05 & \cmark & \xmark &      27 & libgdx    & 30  & 0.12 & \cmark & \cmark \\
          2  & ESearch & 25  & 0.07 & \cmark & \cmark &      28 & libgdx    & 32  & 0.21 & \cmark & \cmark \\
          3  & ESearch & 101 & 0.20 & \cmark & \cmark &      29 & libgdx    & 71  & 0.16 & \cmark & \cmark \\
          4  & ESearch & 63  & 0.49 & \cmark & \cmark &      30 & MPAndroid & 47  & 0.44 & \xmark & \cmark \\
          5  & ESearch & 90  & 4.45 & \cmark & \cmark &      31 & MPAndroid & 66  & 0.17 & \cmark & \cmark \\
          6  & ESearch & 136 & 4.07 & \cmark & \cmark &      32 & MPAndroid & 109 & 0.16 & \cmark & \cmark \\
          7  & ESearch & 15  & 2.09 & \cmark & \cmark &      33 & MPAndroid & 44  & 0.10 & \cmark & \cmark \\
          8  & ESearch & 30  & 0.11 & \xmark & \xmark &      34 & MPAndroid & 62  & 0.16 & \cmark & \cmark \\
          9  & ESearch & 25  & 0.09 & \xmark & \xmark &      35 & MPAndroid & 43  & 0.11 & \cmark & \xmark \\
          10 & ESearch & 21  & 0.15 & \xmark & \xmark &      36 & MPAndroid & 35  & 0.23 & \xmark & \xmark \\
          11 & iosched & 63  & 0.19 & \cmark & \cmark &      37 & MPAndroid & 37  & 0.39 & \xmark & \xmark \\
          12 & iosched & 64  & 0.07 & \cmark & \cmark &      38 & okhttp    & 28  & 0.10 & \xmark & \cmark \\
          13 & kotlin  & 96  & 0.16 & \xmark & \cmark &      39 & retrofit  & 66  & 1.67 & \cmark & \cmark \\
          14 & kotlin  & 54  & 0.57 & \cmark & \cmark &      40 & retrofit  & 78  & 1.76 & \cmark & \cmark \\
          15 & kotlin  & 53  & 0.48 & \cmark & \cmark &      41 & RxJava    & 28  & 0.20 & \cmark & \cmark \\
          16 & kotlin  & 53  & 0.11 & \cmark & \cmark &      42 & spring    & 107 & 0.12 & \cmark & \cmark \\
          17 & kotlin  & 104 & 0.49 & \cmark & \cmark &      43 & spring    & 77  & 0.23 & \xmark & \xmark \\
          18 & kotlin  & 86  & 0.31 & \cmark & \cmark &      44 & spring    & 82  & 0.15 & \cmark & \cmark \\
          19 & kotlin  & 127 & 4.19 & \cmark & \xmark &      45 & spring    & 81  & 0.21 & \cmark & \cmark \\
          20 & kotlin  & 56  & 0.62 & \cmark & \cmark &      46 & spring    & 44  & 0.15 & \cmark & \xmark \\
          21 & kotlin  & 11  & 0.06 & \cmark & \cmark &      47 & spring    & 37  & 0.30 & \cmark & \xmark\\
          22 & kotlin  & 77  & 0.18 & \cmark & \cmark &      48 & spring    & 42  & 0.07 & \cmark & \cmark \\
          23 & kotlin  & 11  & 0.06 & \cmark & \cmark &      49 & spring    & 36  & 0.06 & \cmark & \cmark \\
          24 & kotlin  & 38  & 0.15 & \cmark & \cmark &      50 & spring    & 64  & 0.20 & \xmark & \cmark \\
          25 & kotlin  & 67  & 0.33 & \cmark & \xmark &      51 & spring    & 13  & 0.09 & \xmark & \xmark\\
          26 & kotlin  & 7   & 0.19 & \xmark & \xmark &      52 & spring    & 20  & 0.05 & \xmark & \cmark \\
      \hline
      \end{tabular}
\caption{Summary of experimental results.}
\label{table:expts}
\vspace{-0.2in}
\end{table*}
}

{\small
\begin{table}
      \begin{tabular}{|c | c | c | c |}
      \hline
        \textsc{SafeMerge}  & \texttt{kdiff3}  & Count & Implication \\
      \hline
        \cmark & \cmark & 33 & Verified textual merge \\ 
        \cmark & \xmark & 6 & Verified manual merge \\
        \xmark & \cmark & 5  & Fail to verify textual merge \\
        \xmark & \xmark & 8  & Fail to verify manual merge \\
      \hline
      \end{tabular}
\caption{Summary of differences between \textsc{SafeMerge} and \texttt{\small kdiff3}. 
``Count'' denotes the number of instances in Table~\ref{table:expts}.}
\label{table:expts-summary}
\vspace{-0.2in}
\end{table}
}

\vspace{0.1in} \noindent
{\bf \emph{Main results.}} The results of our  evaluation are presented in~\Cref{table:expts}. 
For each benchmark,~\Cref{table:expts} shows the abbreviated name of 
the application it is taken from (column ``App''), the number of lines
of code in the merge candidate (``LOC''), the running time 
of \toolname \ in seconds (``Time''), and the results produced 
by \toolname \ and \texttt{\small kdiff3}.
Specifically, for \toolname, a checkmark (\cmark) indicates that 
it \ was able to verify semantic conflict-freedom, whereas \xmark\ means
that it produced a warning.
In the case of \texttt{\small kdiff3}, a checkmark indicates the 
absence of \emph{syntactic} conflicts. 

As we can see from~\Cref{table:expts}, \toolname\ is able to verify 
semantic conflict-freedom for 39 of the 52 benchmarks and
reports a warning for the remaining 13. 
We manually inspected these thirteen benchmarks and found eleven instances 
of an actual semantic conflict (i.e., the merge candidate is indeed 
incorrect with respect to~\Cref{def:semantic}).
The remaining two warnings are false positives caused by 
imprecision in the dependence analysis and modeling of collections.
In all, these results indicate that \toolname\ is quite precise, with
a false positive rate around 15\%. 
Furthermore, this experiment also corroborates that \toolname\ is 
practical, taking an average of 0.5~second to verify each benchmark.

Next,~\Cref{table:expts-summary} compares the results produced 
by \toolname\ and \texttt{\small kdiff3} on the 52 benchmarks used 
in our evaluation.
This comparison is very relevant because the merge candidate in these 
benchmarks matches exactly the merge produced by \texttt{\small kdiff3} 
whenever it does not report a textual conflict.
As shown in~\Cref{table:expts-summary}, 33 benchmarks 
are classified as conflict-free by both \toolname\ and \texttt{\small kdiff3}, 
meaning that \toolname\ can verify the correctness of the textual 
merge generated by \texttt{\small kdiff3} in these cases.
For instance, the merge with ID 41 in~\Cref{table:expts} corresponds precisely to the example from RxJava present in~\Cref{sec:overview} (~\Cref{fig:running-example}).
Perhaps more interestingly, we find five benchmarks for 
which \texttt{\small kdiff3}  generates a textual 
merge that is semantically incorrect according to \toolname. 
Among these five instances, two correspond (with IDs 13, 30) to the false 
positives discussed earlier, leaving us with three benchmarks 
where  the merge generated by \texttt{\small kdiff3} violates~\Cref{def:semantic} and should be further investigated by the developers.  
%

As we can see from~\Cref{table:expts-summary}, there are fourteen 
benchmarks that are syntactically conflicting according to \texttt{\small kdiff3} and 
were likely resolved manually by a developer. 
Among these, \toolname\ can verify the correctness of the merge candidate 
for six instances (spread over four different applications), thereby confirming the existence of real-world scenarios 
where syntactic conflict-freedom results in false positives. 
Finally, there are eight cases where the manual merge cannot be 
verified \toolname. 
While these examples indeed violate semantic conflict-freedom, they do 
not necessarily correspond to bugs (e.g., a developer might have 
intentionally discarded changes made by another developer).
For example, in the merge with ID 36 from~\Cref{table:expts}, both variants A and B weaken a predicate in two different ways by adding two and one additional disjuncts respectively\footnote{Merge commit \url{https://github.com/PhilJay/MPAndroidChart/commit/9531ba69895cd64fce48038ffd8df2543eeea1d2 }}.
However, the merge M only picks the weaker predicate from A, thereby effectively discarding some of the changes from variant B. 

\subsection{Evaluation on Erroneous Merge Candidates} 

\begin{table*}[!t]
\small
\begin{center}
\noindent\begin{tabular}{|l|l|c|c|}
\hline
\textbf{Name} & {\textbf{\ \ \ Description}} & \textbf{Time (s)} & {\bf Result}  \\
\hline
\rm\sc B1-kdiff3 & \ Patch gets duplicated in merge & 0.36 & \xmark \\
\rm\sc B1-manual & \ Correct version of above       & 0.38 & \cmark \\
\hline
\rm\sc B2-kdiff3 & \ Semantically same, syntactically different patches & 0.42 & \xmark \\
\rm\sc B2-manual & \  Correct version of above & 0.33 & \cmark \\
\hline
\rm\sc B3-kdiff3 & \ Inconsistent changes in assignment (conflict) & 0.34 & \xmark \\
\hline
\rm\sc B4-kdiff3 & \ Interference between refactoring and insertion (conflict) & 0.31 & \xmark \\
\hline
\rm\sc B5-kdiff3 & \ Interference between insertion and deletion (conflict)  & 0.30 & \xmark \\
\hline
\rm\sc B6-kdiff3 & \ One patch supercedes the other  & 0.32 & \xmark \\
\rm\sc B6-manual & \ Correct version of above   & 0.29 &  \cmark \\
\hline
\rm\sc B7-kdiff3 & \ Inconsistent patches due to off-by-one error (conflict) & 0.29 & \xmark \\
\hline
\end{tabular}
\end{center}
\caption{Results of our evaluation on merges generated by {\tt \small kdiff3}.}
\label{fig:buggy}
\vspace{-0.15in}
\end{table*}

In our second experiment, we explore whether \toolname\ is able to pinpoint 
erroneous merges generated by {\tt \small kdiff3}.
To perform this experiment, we consider base program with ID~$=25$ from~\Cref{table:expts} and 
generate variants by performing various kinds of mutations to the base program. 
Specifically, we design pairs of mutations that cause {\tt \small kdiff3} to generate buggy merge candidates. 

The results of this experiment are summarized in~\Cref{fig:buggy}, where the column 
labeled ``Description" summarizes the nature of the mutation. 
For each pair of variants that are semantically conflict-free, the version 
named {\sc -kdiff3} shows the incorrect merge generated by {\tt \small kdiff3}, 
where as the one labeled {\sc -manual} shows the correct merge that we 
generated manually.
For benchmarks that are semantically conflicting, we only provide results 
for the incorrect merge generated by {\tt \small kdiff3} since a correct 
merge simply does not exist.

The results from~\Cref{fig:buggy} complement those from~\Cref{subs:github}
and provide further evidence that  a widely-used merge tool 
like {\tt \small kdiff3}  can generate erroneous merges and 
that these buggy merges can be detected by our proposed technique.
This experiment  also demonstrates that \toolname \ can verify 
conflict-freedom in the manually constructed correct merges.
%

\subsection{Evaluation of Scalability and Design Choices} 
To assess the scalability of the proposed technique, we performed a third 
experiment in which we compare the running time of \toolname\ against the 
number of lines of code and number of edits. 
To perform this experiment, we start with an existing benchmark from 
the \toolname\  test suite and increase the number of lines of code 
using loop unrolling. 
We also vary the number of edits by injecting a modification in the loop body. 
This way, the number of holes in the shared program increases with each loop unrolling. 

\begin{figure}[!t]
\centering
\includegraphics[scale=0.42]{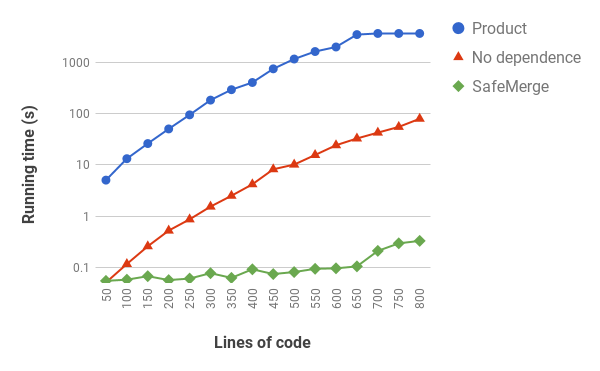}
\vspace{-0.2in}
\caption{Lines of code vs. running time. }
\label{fig:scale-loc}
\vspace{-0.15in}
\end{figure}
 
To evaluate the benefits of the various design choices that we adopt in this paper, we also compare \toolname\ with two variants of itself. In one variant, namely \emph{Product}, we model the shared program using a single hole, so each edit corresponds to one of the program versions. Essentially, this method computes the product of the four program versions using the rules from~\Cref{fig:product} and allows us to assess the benefits of representing program versions as edits applied to a shared program. In another variant called \emph{No dependence}, we do not abstract away shared program fragments using uninterpreted functions and analyze them by constructing a 4-way product. However, we still combine reasoning from different product programs in a compositional way.
 
\Cref{fig:scale-loc} compares the running time of \toolname\ against these two variants as we vary the number of lines of code but \emph{not} the number of edits. Observe that the y-axis is shown in log scale. As we can see from this plot, \toolname\ scales quite well and analyzes each benchmark in under a second. In contrast, the running time of \emph{Product} grows exponentially in the lines of code.
As expected, the \emph{No dependence} variant is better than \emph{Product} but significantly worse than \toolname. 

\begin{figure}[!t]
\centering
\includegraphics[scale=0.44]{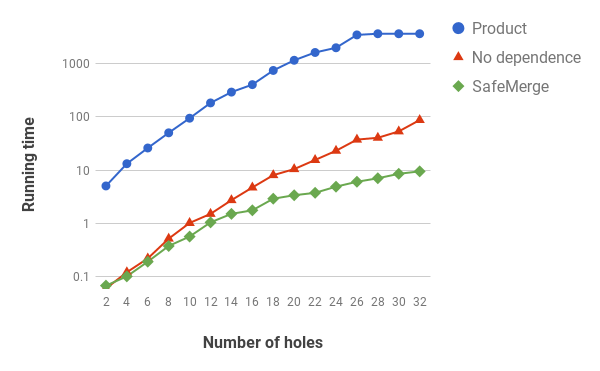}
\vspace{-0.2in}
\caption{Number of holes (edits) vs. running time. Lines of code varies between 50 and 800.}
\label{fig:scale-holes}
\vspace{-0.2in}
\end{figure}
 
Next, Figure~\ref{fig:scale-holes} compares the running time of \toolname\ against \emph{Product} and \emph{No dependence} as we vary \emph{both} the number of lines of code and the number of edits. Specifically, a benchmark containing $n$ holes contains $25n$ lines of code, and the $y$-axis shows the running time of each variant in log scale. As expected, \toolname \ is more sensitive to the number holes than it is to the number of lines of code because it abstracts away shared program fragments. However, \toolname  \ still significantly outperforms both \emph{Product} and \emph{No Dependence}. In particular, for a program with 32 edits and 800 lines of code, \toolname \ can verify semantic conflict freedom in approximately 10 seconds, while \emph{No Dependence} takes approximately 100 seconds and \emph{Product} times out.
 
In summary, this experiment shows that \toolname \ scales well as we vary the lines of code and that its running time is still feasible when program variants perform over 30 modifications to the base program in this example. This experiments also corroborates the practical importance of representing program versions as edits applied to a shared program as well as the advantage of abstracting away shared program fragments using uninterpreted functions.

\section{Related Work}
\label{section:related}
In this section, we compare our technique with prior work on program merging and  relational verification.

\noindent \vspace{-0.25in}
\paragraph{{\bf \emph{Structure-aware merge.}}} Most algorithms for program merging are textual in nature, hardly ever formally described~\cite{diff3-pierce}, and without
semantic guarantees. To improve on this situation, previous work has proposed  {\em structured} and {\em semi-structured} merge techniques to better resolve merge conflicts.
For example,  {\tt \small FSTMerge}~\cite{Apel11} uses syntactic structure to resolve conflicts between AST nodes
that can be reordered (such as method definitions), but it falls back on unstructured textual merge for other kinds of nodes. Follow-up work on {\tt \small JDime}~\cite{Apel12,Lebetaenich15} improves  the poor performance of structure-based merging by using textual-based mode (fast)  as long as no conflicts are detected, but switches to structure-based mode 
in the presence of conflicts. However, none of these techniques guarantee semantic conflict freedom.

\noindent 
\vspace{-0.25in}
\paragraph{{\bf \emph{Semantics-aware merge.}}}
Our work is inspired by earlier work on \emph{program integration}, which originated with the HPR algorithm~\cite{hpr} for checking \emph{non-interference} and generating valid merges. The HPR algorithm  was later refined by the  work of Yang et al.~\cite{Yang90}, which is one of the first attempts to incorporate semantics for merge generation.
In that context, the notion of conflict-freedom is parameterized by a classification of nodes of the variants as {\it unchanged} such that the backward slices of unchanged nodes in the two variants are equivalent modulo
a semantic correspondence. Thus, their classification algorithm is  parameterized by a semantic \emph{congruence} relation. 
Our approach tackles the slightly different  \emph{merge verification} (rather then \emph{merge generation}) problem, but
improves on these prior techniques in several dimensions: 
First,  we do not require annotations to map statements across the different versions --- this information is  computed automatically using our edit generation algorithm (Sec~\ref{sec:edit-gen}). 
Second, we show how to formulate conflict freedom  directly with verification conditions and assertion checking.
Finally, our approach performs precise, compositional reasoning about  edits by combining lightweight dependence analysis with relational reasoning using product programs. 

\noindent 
\vspace{-0.25in}
\paragraph{{\bf \emph{Relational verification.}}} Verification of conflict freedom is related to a  line of work on relational program logics~\cite{rhl,rsl,chl} and product programs~\cite{product1,product2,covac}. For instance,  Benton's Relational Hoare Logic (RHL)~\cite{rhl} allows proving equivalence between a pair of structurally similar programs. Sousa and Dillig generalize Benton's work by developing Cartesian Hoare Logic, which is used for proving $k$-safety of programs~\cite{chl}. Barthe et al. propose another technique for relational verification using \emph{product programs}~\cite{product1,product2} and apply their technique to relational properties, such as equivalence and 2-safety~\cite{2safety}.  In this work, we build on the notion of product programs used in prior work~\cite{product1,product2,covac}. However, rather than constructing a monolithic product of the four program version, we construct \emph{mini-products} for each edit. Furthermore, our proposed product construction algorithm differs from prior techniques in that it uses similarity metrics to guide synchronization and generalizes to $n$-way products.

\noindent 
\vspace{-0.25in}
\paragraph{{\bf \emph{Cross-version program analysis.}}} There has been renewed interest in program analysis techniques for answering questions about program differences across versions~\cite{LahiriFoser10}. 
Prior work on comparing closely related programs versions include {\it regression verification} that checks semantic equivalence using uninterpreted function abstraction of equivalent callees~\cite{Godlin08,Lahiri12,felsing-ase14}, mutual summaries~\cite{HawblitzelCade13,wood-esop17}, relational invariant inference to prove differential properties~\cite{LahiriFse13} and verification modulo versions~\cite{LogozzoPldi14}. 
Other approaches include static analysis for abstract differencing~\cite{Jackson94,PartushOOPSLA14}, symbolic execution for verifying assertion-equivalence~\cite{ramos-cav11} and {\it differential symbolic execution} to summarize differences~\cite{PersonFse08}. 
Our work is perhaps closest to differential assertion checking~\cite{LahiriFse13} in the use of product programs and  invariant inference. 
However, we do not require an assertion and verify a more complex property involving four different programs.
We note that bugs arising from 3-way merges could potentially also be uncovered using {\em multi-version testing}~\cite{varan:asplos15}.

\section{Conclusion and Future Work}
\label{section:conc}

We have proposed a notion of \emph{semantic conflict freedom} for 
3-way merges and described a  verification algorithm for proving this property.
Our verification algorithm analyzes the edited parts of the program in a precise way using product programs, but leverages lightweight dependence analysis to reason about program fragments that are shared between all program versions. 
Our evaluation shows that the proposed approach can verify semantic conflict-freedom for many real-world benchmarks and identify issues in problematic merges that are generated by textual 3-way merge tools.

We view this work as a first step towards precise, 
semantics-aware \emph{merge synthesis}.  In future work, we plan to explore synthesis techniques that can automatically generate correct-by-construction 3-way program merges. Since correct merge candidates should obey semantic conflict freedom, the verification algorithm proposed in this paper is necessarily a key ingredient of such semantics-aware merge synthesis tools.

\bibliography{progmerge}
\ifsupp
\newpage \ \newpage

\subsection*{Appendix A: Operational Semantics}

\begin{figure}
\[
\begin{array}{cc}
\irule{
}{\sigma \vdash {\rm skip} \Downarrow \sigma } & 
\irule{
\begin{array}{cc}
\sigma \vdash e \Downarrow c & \sigma' = \sigma[(x ,0) \mapsto c]
\end{array}
}{\sigma \vdash x := e\Downarrow \sigma' } 
\\ \ \\ 
\irule{
\begin{array}{c}
\sigma \vdash e_1 \Downarrow c_1  \ \ \sigma \vdash e_2 \Downarrow c_2 \\
 \sigma' = \sigma[(x ,c_1) \mapsto c_2]
\end{array}
}{\sigma \vdash x[e_1] := e_2\Downarrow \sigma'} &
\irule{
\begin{array}{c}
\sigma[(x,0)] = c
\end{array}
}{\sigma \vdash out(x)\Downarrow \sigma } \\ \ \\ 
\irule{
\begin{array}{c}
\sigma \vdash S_1 \Downarrow \sigma_1 \\
\sigma_1 \vdash S_2 \Downarrow \sigma_2
\end{array}
}{\sigma \vdash S_1; S_2 \Downarrow \sigma_2 }  &
\irule{
\begin{array}{c}
\sigma \vdash C \Downarrow {\rm true} \\
\sigma \vdash S_1 \Downarrow \sigma_1 \\
\end{array}
}{\sigma \vdash \ifstmt{C}{S_1}{S_2} \Downarrow \sigma_1 } \\ \ \\ 
\irule{
\begin{array}{c}
\sigma \vdash C \Downarrow {\rm false} \\
\sigma \vdash S_2 \Downarrow \sigma_2 \\
\end{array}
}{\sigma \vdash \ifstmt{C}{S_1}{S_2} \Downarrow \sigma_2 } &
\irule{
\begin{array}{c}
\sigma \vdash C \Downarrow {\rm false} 
\end{array}
}{\sigma \vdash \while{C}{S} \Downarrow \sigma } \\ \ \\ 
\irule{
\begin{array}{c}
\sigma \vdash C \Downarrow {\rm true} \\
\sigma \vdash S \Downarrow \sigma_1 \\
\sigma_1 \vdash  \while{C}{S} \Downarrow \sigma_2
\end{array}
}{\sigma \vdash \while{C}{S} \Downarrow \sigma_2 } 
\end{array}
\]
\caption{Operational semantics}\label{fig:semantics}
\end{figure}

Figure~\ref{fig:semantics} shows the operational semantics of the language from Figure~\ref{fig:lang}. Recall that $\sigma$ maps (variable, index) pairs to values, and we view scalar variables as arrays with a single valid index at 0. Since the semantics of expressions is completely standard, we do not show them here. However, one important point worth noting is the semantics of expressions involving array reads:
\[
\begin{array}{cc}
\irule{
\begin{array}{c}
\sigma \vdash e \Downarrow c \\
(a, c) \in dom(\sigma) 
\end{array}
}{\sigma \vdash a[e] \Downarrow \sigma[(a, c)]} & 
\irule{
\begin{array}{c}
\sigma \vdash e \Downarrow c \\
(a, c) \not \in dom(\sigma) 
\end{array}
}{\sigma \vdash a[e] \Downarrow \bot}  
\end{array}
\]
In other words, reads from locations that have not been initialized yield a special constant $\bot$.

\subsection*{Appendix B: Soundness of Product}

Here, we provide a proof of~\Cref{thm:product-sound}. 
The proof is by structural induction over the product construction 
rules given in~\Cref{fig:product}.
Since the two directions of the proof are completely symmetric, we 
only prove one direction. 
Note that the base case is trivial because $\vdash \stmt \gen \stmt$.

\paragraph{Rule 1.} Suppose $\sigma \vdash A \Downarrow \sigma'$ and 
$\sigma' \vdash \stmt_1; \stmt_2; \ldots; \stmt_n \Downarrow \sigma''$.  
By the premise of the proof rule and the inductive hypothesis, we have 
$\sigma' \vdash \stmt \Downarrow \sigma''$.
Thus, $\sigma \vdash A; \stmt \Downarrow \sigma''$.

\paragraph{Rule 2.} Suppose 
$\sigma \vdash (\ifstmt{C}{\stmt_{t}}{\stmt_{e}}); \stmt_1; \stmt_2; \ldots; \stmt_n \Downarrow \sigma''$. 
Without loss of generality, suppose $\sigma \vdash C \Downarrow true$, 
and suppose $\sigma \vdash \stmt_{t};\stmt_1 \Downarrow \sigma'$, 
so $\sigma' \vdash \stmt_2; \ldots; \stmt_n \Downarrow \sigma''$. 
By the first premise of the proof rule and the inductive hypothesis, 
we have $\sigma \vdash \stmt' \Downarrow \sigma''$. 
%
%
%
Hence, $\sigma \vdash \ifstmt{C}{\stmt'}{\stmt''} \Downarrow \sigma''$. 

\paragraph{Rule 3.} 
Let~$\stmt_x = \while{C_1}{\stmt_{B_1}}; \stmt_1$.
Suppose we have 
$\sigma \vdash \stmt_x; \stmt_2; \ldots; \stmt_n \Downarrow \sigma'$. 
Suppose there is exists~$\stmt_i$ that satisfies first premise of the proof rule.
Observe that $\stmt_1; \stmt_2; \ldots; \stmt_n;$ is semantically equivalent to 
$\stmt_n; \stmt_2; \ldots; \stmt_{n-1}; \stmt_1;$ as long as $\stmt_1$, $\stmt_n$, 
do not share variables between them and also with~$\stmt_2 \ldots \stmt_{n-1}$. 
Since~$\stmt_x$ and $\stmt_i$ have no shared variables between them and with any
other program~$\stmt_j$ different than~$\stmt_x$ and~$\stmt_i$, we have  

\[ \sigma \vdash  \stmt_i; \stmt_2; \ldots; \stmt_{i-1};\stmt_{i+1}; \ldots \stmt_n; \while{C_1}{\stmt_{B_1}}; \stmt_1 \Downarrow \sigma' \] 

Then, by the premise of the proof rule and the inductive hypothesis, we have $\sigma \vdash \stmt \Downarrow \sigma'$.

\paragraph{Rule 4.} Suppose we have 
$\sigma \vdash \stmt_1; \stmt_2; \ldots; \stmt_n \Downarrow \sigma''$ 
where each~$\stmt_i$ is of the form~$\while{C_i}{\stmt_{B_i}}; \stmt_i'$.
By the same reason as in Rule 3. we can move any loop in each~$\stmt_i$ 
to the beginning as they don't share any variable with any other~$\stmt_j$.
That is, considering~$H = \stmt_1; \ldots; \stmt_o$ be the set of programs
satisfying the second premise we have

\[ \sigma \vdash  \while{C_1}{\stmt_{B_1}}; \ldots \while{C_o}{\stmt_{B_o}}  \Downarrow \sigma' \] 

and considering $\stmt_{o+1}; \ldots \stmt{n}$ a sequence of the original programs excluding
the ones in~$H$ we have

\[ \sigma' \vdash  \stmt_1'; \ldots \stmt_o'; \stmt_{o+1}; \ldots \stmt{n} \Downarrow \sigma'' \] 

Then, by the last premises of the proof rule and the inductive hypothesis, 
we have that~$\sigma \vdash \stmt';\stmt'' \Downarrow \sigma''$.

\paragraph{Rule 5.} Suppose we have \[ \sigma \vdash \while{C_1}{\stmt_1}; \while{C_2}{\stmt_2}; \stmt_3; \ldots; \stmt_n \Downarrow \sigma' \] Let $W'$ be the loop $\while{C_1 \land C_2}{\stmt_1; \stmt_2}$. Since $C_1, C_2$ and  $\stmt_1, \stmt_2$ have disjoint sets of variables, the program fragment $\while{C_1}{\stmt_1}; \while{C_2}{\stmt_2}$ is semantically equivalent to $W'; R$ (where $R$ comes from the third line of the proof rule). Hence, we have $\sigma \vdash W'; R; \stmt_3; \ldots; \stmt_n \Downarrow \sigma' $. By the first premise of the proof rule and the inductive hypothesis, if $\sigma_0 \vdash \stmt_1; \stmt_2 \Downarrow \sigma_1$ for any $\sigma_0, \sigma_1$, then $\sigma_0 \vdash \stmt \Downarrow \sigma_1$. Thus, $\sigma \vdash W' \Downarrow \sigma^*$ implies $\sigma \vdash W \Downarrow \sigma^*$, which in turn implies $\sigma \vdash W; R; \stmt_3; \ldots; \stmt_n \Downarrow \sigma'$. By the last premise of the proof rule and the inductive hypothesis, we know $\sigma \vdash \stmt': \sigma'$; hence, the property holds.

\subsection*{Appendix C: Proof of Soundness of Relational Post-conditions}

The proof is  by structural induction on $\hstmt$.

\paragraph{Case 1.} $\hstmt = \hole$, and the edits are $\stmt_1, \ldots, \stmt_4$. In this case, Figure~\ref{fig:relpost} constructs the relational post-condition by first computing the product program $\stmt$ as $\stmt_1[V_1/V] \crp \ldots \crp \stmt[V_4/4]$ and then computing the standard post-condition of $\stmt$. By Theorem~\ref{thm:product-sound}, we have $\sigma \vdash \stmt \Downarrow \sigma'$ iff $\sigma \vdash \stmt_1[V_1/V] \crp \ldots \crp \stmt[V_4/4] \Downarrow \sigma'$. Furthermore, by the correctness of \emph{post} operator, we know that $\{\cond \} \stmt \{ \cond'\}$ is a valid Hoare triple. This implies $\{ \cond \} \stmt_1[V_1/V]; \ldots; \stmt_4[V_4/V] \{\cond'\}$ is also a valid Hoare triple.
\paragraph{Case 2.} $\hstmt = \stmt$ (i.e., $\hstmt$ does not contain holes). 
By the second rule in Figure~\ref{fig:relpost}, we know that $\{ \cond \} \stmt_1; \ldots; \stmt_n \{ \cond'\} $ is a valid Hoare triple. Now, consider any valuation $\sigma$ satisfying $\cond$. By the correctness of the Hoare triple,  if $\sigma \vdash  \stmt_1; \ldots; \stmt_n \Downarrow \sigma'$, we know that $\sigma'$ also satisfies  $\cond'$. Now, recall that $\stmt_1; \ldots; \stmt_n$ contains uninterpreted functions, and we assume that $F(\vec{x})$ can return any value, as long as it returns something consistent for the same input values. Let $\Sigma$ represent the set of all valuations $\sigma_i$ such that $\sigma \vdash \stmt_1; \ldots; \stmt_n \Downarrow \sigma_i$. By the correctness of the Hoare triple, we know that any $\sigma_i \in \Sigma$ satisfies $\cond'$. Assuming the correctness of the mod and dependence analysis, for any valuation $\sigma$ such that $\sigma \vdash \stmt[V_1/V]; \ldots; \stmt[V_4/V] \Downarrow \sigma'$, we know that $\sigma' \in \Sigma$. Since all valuations in $\Sigma$ satisfy $\cond'$, this implies $\sigma'$ also satisfies $\cond'$. Thus, $\{\cond\}\stmt[V_1/V]; \ldots; \stmt[V_4/V] \{\cond'\}$ is also a valid Hoare triple.

\paragraph{Case 3.} $\hstmt = \hstmt_1; \hstmt_2$. Let $\vec{\edit}_A$ denote the prefix of $\vec{\edit}$ that is used for filling holes in $\hstmt_1$, and $\vec{\edit}_B$ denote the prefix of $\vec{\edit}_1$ that is used for filling holes in $\hstmt_2$.  By the premise of the third rule and inductive hypothesis, we have \[ \{\cond\} (\hstmt_1[\edit_{A1}])[V_1/V]); \ldots; (\hstmt_1[{\edit}_{A4}])[V_4/V]) \{\cond_1\}\] as well as 
 \[  \{\cond_1\} (\hstmt_2[\edit_{B1}])[V_1/V]); \ldots; (\hstmt_2[\edit_{B4}])[V_4/V]) \{\cond_2\} \]
 Using these and the standard Hoare rule for composition, we can conclude:
 \[
 \begin{array}{lll}
  \{\cond\}  & (\hstmt_1[\edit_{A1}])[V_1/V]); \ldots; (\hstmt_1[\edit_{A4}])[V_4/V]);  & \\
   & (\hstmt_2[\edit_{B1}])[V_1/V]); \ldots; (\hstmt_2[\edit_{B4}])[V_4/V]) & \{\cond_2\}
  \end{array}
  \] 
 Since we can commute statements over different variables, this implies:
  \[
  \begin{array}{lll}
  \{\cond\} & (\hstmt_1[\edit_{A1}]; \hstmt_2[\edit_{B1}])[V_1/V]); \ldots; &
  \\
  & (\hstmt_1[\edit_{A4}]; \hstmt_2[\edit_{B4}])[V_4/V]) & 
  \{\cond_2\}
  \end{array}
  \] 
 Next, using the fact that $\hstmt[\edit_i] = (\hstmt_1; \hstmt_2)[\edit_{i}] = \hstmt_1[{\edit_{Ai}}]; \hstmt_2[{\edit_{Bi}}]$, we can conclude:
  \[ \{\cond\} (\hstmt[{\edit_1}])[V_1/V]); \ldots; (\hstmt[{\edit_4}])[V_4/V]) \{\cond_2\}\] 
  
  \paragraph{Case 4.} $\hstmt = \ifstmt{C}{\hstmt_1}{\hstmt_2}$. Let $\vec{\edit}_A, \vec{\edit}_B$ denote the prefixes of $\vec{\edit}, \vec{\edit_1}$ that is used for filling holes in $\hstmt_1$ and $\hstmt_2$ respectively. Also, let $C_i$ denote $C[V_i/V]$. By the first premise of rule 4 from Figure~\ref{fig:relpost} and the inductive hypothesis, we have:
 \[
  \{\cond \land C_1 \} \ \ (\hstmt_1[\edit_{A1}])[V_1/V]); \ldots; (\hstmt_1[\edit_{A4}])[V_4/V]) \ \ 
  \{\cond_1\}\] 
  Now, using the second premise and the inductive hypothesis, we also have:
   \[
  \{\cond \land \neg C_1 \} \ \ (\hstmt_2[\edit_{B1}])[V_1/V]); \ldots; (\hstmt_2[\edit_{B4}])[V_4/V]) \ \ 
  \{\cond_2\}\] 
  Using these two facts and the standard Hoare logic rule for if statements, we get:
  
 \[
 \begin{array}{lll}
  \{\cond \} &  C_1  ? {(\hstmt_1[\edit_{A1}])[V_1/V]); \ldots; (\hstmt_1[\edit_{A4}])[V_4/V])} : &\\ 
  & {(\hstmt_2[\edit_{B1}])[V_1/V]); \ldots; (\hstmt_2[\edit_{B4}])[V_4/V])} &
  \{\cond_1\}
  \end{array}\] 
  Now, since $\varphi$ logically entails $\bigwedge_{i,j} C_i \leftrightarrow C_j$, the statement above is equivalent to:
  \[
  \begin{array}{l}
  \ifstmt{C_1}{(\hstmt_1[\edit_{A1}])[V_1/V]}{(\hstmt_2[\edit_{B1}])[V_1/V]}; \\
  \ldots \\
   \ifstmt{C_4}{(\hstmt_1[\edit_{A4}])[V_4/V]}{(\hstmt_2[\edit_{B4}])[V_4/V]}; \\
  \end{array}
  \]
  
   Next, using the fact that $\hstmt[\edit_i] = (\ifstmt{C}{\hstmt_1}{\hstmt_2})[\edit_i] = \ifstmt{C}{\hstmt_1[{\edit}_{Ai}]}{\hstmt_2[\vec{\edit}_{Bi}]}$, we can conclude:
     \[
     \begin{array}{lll}
     \{ \cond \} &
  ((\ifstmt{C_1}{\hstmt_1}{\hstmt_2})[\edit_1])[V_1/V]; \ldots; & \\
  & ((\ifstmt{C_4}{\hstmt_1}{\hstmt_2}))[\edit_4])[V_4/V] &
 \{ \cond'\}
 \end{array}
  \]
   
  \paragraph{Case 5.} $\hstmt = \while{C}{\hstmt}$. As in case (4), let $C_i$ denote $C[V_i/V]$. From the premise of rule (5) of Figure~\ref{fig:relpost} and the inductive hypothesis, we know: 
  \[
  \{ \inv \}\ \  (\hstmt[\edit_1])[V_1/V]; \ldots (\hstmt[\edit_4])[V_4/V] \ \ \{\inv \}
  \]
  Since we also have $\varphi \models \inv$ from the premise, this implies:
  \[
  \{ \cond \}\ \  \while{C_1}{(\hstmt[\edit_1])[V_1/V]; \ldots (\hstmt[\edit_4])[V_4/V]} \ \ \{\inv \land \neg C_1\}
  \]
  Next, since we can commute statements over different variables and $\inv$ implies $\bigwedge_{ij} C[V_i/V] \leftrightarrow C[V_j/V] $, we can conclude:
  \[
  \begin{array}{lll}
  \{ \cond \} &  \while{C_1}{(\hstmt[\edit_1])[V_1/V]}; \ldots ; &\\
  & \while{C_4}{(\hstmt[\edit_4])[V_4/V]} & \{\inv \land \neg C_1\}
  \end{array}
  \]
 Finally, because  the loop $\while{C_i}{\hstmt[\edit_i]}$ is the same as $(\while{C_i}{\hstmt})[\edit_i]$, we have:

    \[
    \begin{array}{lll}
  \{ \cond \} &  (\while{C_1}{\hstmt[V_1/V]})[\edit_1]; \ldots ; &\\
  & (\while{C_4}{\hstmt[V_4/V]})[\edit_4] &
  \{\inv \land \neg C_1\}
  \end{array}
  \]
  
    \paragraph{Case 6.} First, assuming the soundness of the standard \emph{post} operator, we have $\{ \cond \} \stmt \{ post(\stmt, \cond) \}$. Using the premise of the proof rule and Theorem~\ref{thm:product-sound}, we obtain:
    \[
    \{ \cond \} \ \ (\hstmt[\edit_1^1])[V_1/V]; \ldots (\hstmt[\edit_4^1])[V_4/V] \ \ \{post(\stmt, \cond) \}
     \]
     
     Since $\edit_i^1$ is the prefix of $\edit_i$  that contains as many holes as $\hstmt$, we also know $\hstmt[\edit_i^1] = \hstmt[\edit_i]$. Thus, we get:
         \[
    \{ \cond \} \ \ (\hstmt[\edit_1])[V_1/V]; \ldots (\hstmt[\edit_4])[V_4/V] \ \ \{post(\stmt, \cond) \}
     \]



\subsection*{Appendix D: Soundness of $n$-way Diff Algorithm}

Theorem~\ref{thm:ndiff} follows directly from the following two lemmas:

\begin{lemma}
If $|\edit|  = \numHoles{\hedit}$, then $\Compose$ ensures the following post-conditions:
\begin{itemize}
\item $|\edit'| = |\hedit|$
\item  For any $\hstmt$ s.t. $\numHoles{\hstmt} = |\hedit|$, $(\hstmt[\hedit])[\edit] = \hstmt[\edit']$
\end{itemize}

\end{lemma}
\begin{proof}
Consider the two postconditions of $\Compose$. 
For the branch $\hedit = [ \ ]$, it is easy to see that $\edit' = \hedit = [ \ ]$ and thus $|\edit'| = |\hedit|$.
For any $\hstmt$ with 0 holes, applying any edits gets back $\hstmt$, satisfying the second postcondition.

For the branch  $\head{\hedit} = \hole$, we know $\numHoles{\tail{\hedit}} = \numHoles{\hedit} - 1 = |\tail{\edit}|$ (given the precondition), which satisfies the precondition of $\Compose$ at line~\ref{line:callCompose1}.
The first postcondition of the recursive call to $\Compose$ implies that size of the return value ($|\edit'|$) equals $|\head{\edit}| + |\tail{\hedit}| = 1 + |\hedit| - 1 = |\hedit|$. 
Now consider a $\hstmt$ such that $\numHoles{\hstmt} = |\hedit|$.
Let $\edit''$ be the return from the recursive call to $\Compose$.
Then $\hstmt[\edit'] = \hstmt[\head{\edit} :: \edit''] = (\hstmt[\head{\edit}])[\edit'']$ (by definition of applying an edit).
Since $\numHoles{\hstmt[\head{\edit}]} = \numHoles{\hstmt} -1 = |\tail{\hedit}|$, we know that $(\hstmt[\head{\edit}])[\edit''] = ((\hstmt[\head{\edit}])[\tail{\hedit}])[\tail{\edit}]$ (from the second postcondition of the recursive call).
Since $\head{\hedit} = \hole$ in this branch, $(\hstmt[\head{\edit}])[\tail{\hedit}] = (\hstmt[\hole :: \tail{\hedit}])[\head{\edit}] = (\hstmt[\hedit])[\head{\edit}]$.
This follows from the fact that applying $\head{\edit}$ to the first hole in $\hstmt$ followed by applying $\tail{\hedit}$ is identical to applying a hole in the first hole in $\hstmt$ followed by applying $\tail{\hedit}$, followed by applying $\head{\edit}$ which applies it to the first hole in $\hstmt$. 
Further, $((\hstmt[\hedit])[\head{\edit}])[\tail{\edit}] = (\hstmt[\hedit])[\edit]$ by the rule of applying edits, which proves this postcondition.

For the branch $\head{\hedit} \neq \hole$, we know $\numHoles{\tail{\hedit}} = \numHoles{\hedit}$. 
This along with the precondition of $\Compose$ establishes the preconditon to the call to $\Compose$ at line~\ref{line:callCompose2}.
Let $\edit''$ denote the return of the recursive call to $\Compose$. 
The recursive call ensures that $|\edit''| = |\tail{\hedit}| = |\hedit| - 1$. Thus $|\edit'| = |\head{\hedit} :: \edit''| = |\hedit|$, which establishes the first postcondition. 
Now consider a $\hstmt$ such that $\numHoles{\hstmt} = |\hedit|$.
Then $\hstmt[\edit'] = \hstmt[\head{\hedit} :: \edit''] = (\hstmt[\head{\hedit}])[\edit'']$.
Since $\numHoles{\hstmt[\head{\hedit}]} = \numHoles{\hstmt} -1 = |\tail{\hedit}|$, we know that $(\hstmt[\head{\hedit}])[\edit''] = ((\hstmt[\head{\hedit])[\tail{\hedit}])[\edit]$ (from the second postcondition of the recursive call), which simplifies to $(\hstmt[\head{\hedit} :: \tail{\hedit}])[\edit] = (\hstmt[\hedit}])[\edit]$ by the property of applying an edit. 
\end{proof}

\begin{lemma}
If $|\edit_i| = \numHoles{\hstmt}$ for all $i \in [1,\ldots,k]$ and \textsf{2Diff} satisfies the contract provided in Algorithm~\ref{alg:twoDiff}, then \textsf{GenEdit} ensures the following post-conditions:
\begin{itemize}
\item $|\edit_{i}'| = \numHoles{\hstmt'}$ for $i \in [1,\ldots,k+1]$ 
\item $\hstmt'[\edit_{k+1}'] = \stmt$ and  $\hstmt'[\edit_i'] = \hstmt[\edit_i]$ for $i \in [1,\ldots,k]$
\end{itemize}
\end{lemma}

\begin{proof}
First, the precondition $|\edit_i| = \numHoles{\hedit}$ of $\Compose$ in line~\ref{line:callCompose} is satisfied from the precondition $|\edit_i| = \numHoles{\hstmt}$ of $\GenEdit$  and the second postcondition $\numHoles{\hedit} = \numHoles{\hstmt}$ of $\DiffTwo$. 

Now, consider the postcondition $|\edit_{i}'| = \numHoles{\hstmt'}$ for $i \in [1,\ldots,k+1]$.
From the first postcondition of $\DiffTwo$ at line~\ref{line:calldiff2}, we know that $\numHoles{\hstmt'} = |\hedit|$.
For any $i \in [1,\ldots,k]$, the first postcondition of $\Compose$ at line~\ref{line:callCompose} implies $|\hedit| = |\edit_i'|$.
Together, they imply that $\numHoles{\hstmt'} = |\edit_i'|$.

The postcondition $\hstmt'[\edit_{k+1}'] = \hstmt$ follows directly from the third postcondition of $\DiffTwo$ at line~\ref{line:calldiff2} and line~\ref{line:last}. 
Now consider $\edit_i'$ for $i \in [1,\ldots,k]$.
We know from the postcondition of $\DiffTwo$ that $\numHoles{\hstmt'} = |\hedit|$. 
Therefore, from the postcondition of $\Compose$ at line~\ref{line:post2Compose} (where we substitute $\hstmt'$ for the bound variable $\hstmt$), we know that $(\hstmt'[\hedit])[\edit_i] = \hstmt'[\edit_i']$.
From the postconditon of $\DiffTwo$ at line~\ref{line:calldiff2}, we know $\hstmt'[\hedit] = \hstmt$.
Together, they imply $\hstmt[\edit_i] = \hstmt'[\edit_i']$. 
\end{proof}

\subsection*{Appendix E: Example of 4-way diff}\label{sec:editGenExample}
We illustrate the 4-way {\it diff} using a simple example:
\[
\begin{array}{lll}
O & \doteq & \ifstmt{c}{x := 1}{y := 2}; z := 3 \\
A & \doteq & \ifstmt{c}{x := 2}{y := 2};  \\
B & \doteq & \ifstmt{c}{x := 1}{y := 3}; z := 3  \\
M & \doteq & \ifstmt{c}{x := 2}{y := 3};   \\
\end{array}
\]
According to Algorithm~\ref{alg:editGen}, we start out with the shared program $\hstmt = O$ and $\edit_\orig = [ \ ]$.

Now consider the first call to $\GenEdit(\hstmt, A, \edit_\orig)$.  
After invoking $\DiffTwo(A, \hstmt)$ at line~\ref{line:calldiff2}, it returns the tuple $(\hstmt', \edit, \hedit)$ where $\hstmt' \doteq \ifstmt{c}{\hole}{y := 2}; \hole$,  $\edit \doteq [x := 2, \skipr]$ and $\hedit \doteq [x := 1, z := 3]$.
The reader can verify that $\hstmt'[\edit] = A$ and $\hstmt'[\hedit] = \hstmt = O$. 
Next, consider the call to $\Compose(\hedit, \edit_1)$ where $\edit_1 = [ \ ]$.
The call executes the branch in line~\ref{line:last} twice (since $\hedit$ does not contain any holes) and returns $\edit'$ as $\hedit$. 
Therefore, the call to $\GenEdit$ returns the tuple $(\hstmt', [x := 1, z := 3], [x := 2, \skipr])$, which constitutes $\hstmt, \edit_\orig, \edit_\vara$ for the next call to $\GenEdit$.

The next call to $\GenEdit(\hstmt, B, \edit_\orig, \edit_\vara)$ calls $\DiffTwo(B, \ifstmt{c}{\hole}{y := 2}; \hole)$ and returns  $(\hstmt', \edit, \hedit)$, where $\hstmt' \doteq \ifstmt{c}{\hole}{\hole}; \hole$ and $\edit \doteq [x := 1, y := 2, z := 3]$ (which becomes $\edit_\varb$) and $\hedit \doteq [\hole, y := 2, \hole]$. 
The reader can verify that $\hstmt'[\edit] = B$ and $\hstmt'[\hedit] = \hstmt$.
The loop at line~\ref{line:loopStart} updates $\edit_\orig$ and $\edit_\vara$ --- we only describe the latter. 
The return of $\Compose(\hedit, \edit_\vara)$ updates $\edit_\vara$ to $[x := 2, y := 2, \skipr]$ by walking the first argument and replacing $\hole$ with corresponding entry from $\edit_\vara$.
Similarly, the $\edit_\orig$ is updated by $\Compose(\hedit, \edit_\orig)$ to $[x := 1, y := 2, z := 3]$.

The final call to $\GenEdit(\hstmt, M, \edit_\orig, \edit_\vara, \edit_\varb)$ returns the tuple $(\hstmt, \edit_\orig, \edit_\vara, \edit_\varb, \edit_\cand)$, where $\hstmt, \edit_\orig, \edit_\vara, \edit_\varb$ remain unchanged (since $\hstmt$ already contains holes at all the changed locations), and $\edit_\cand$ is assigned $[x := 2, y := 3, \skipr]$. 
The reader can verify that $\hstmt[\edit_\orig] = O, \hstmt[\edit_\vara] = A, \hstmt[\edit_\varb] = B, \hstmt[\edit_\cand] = M$.

\subsection*{Appendix F: An abstract implementation of $\DiffTwo$}
\label{sec:two-way-diff}

\begin{algorithm}[htbp]
\caption{Algorithm for 2-way AST Diff}\label{alg:twoDiff}
\begin{algorithmic}[1]
\Procedure{$\DiffTwo$}{$\stmt, \hstmt$}
\State {\rm \bf Input:} A program $\stmt$ and a shared program $\hstmt$
\State {\rm \bf Output:} Shared program $\hstmt'$ and edits $\edit, \hedit$
\State {\rm \bf Ensures:} $|\edit| = |\hedit| = \numHoles{\hstmt'}$ 
\State {\rm \bf Ensures:} $\numHoles{\hedit} = \numHoles{\hstmt}$
\State {\rm \bf Ensures:} $\hstmt'[\edit] = \stmt$, $\hstmt'[\hedit] = \hstmt$ 

\If{$\hstmt = \hole$} \label{line:matchHole}
\Return ($\hole, [\stmt], [\hstmt]$)
\ElsIf {$\hstmt = \stmt$} \label{line:matcheq}
\Return ($\stmt, [\ ], [\ ]$)
\ElsIf {*} \label{line:nondet1}
\Return $\DiffTwo(\stmt;\skipr, \hstmt)$ \Comment{Non-deterministic $\skipr$ \  introduction}
\ElsIf {*} \label{line:nondet2}
\Return $\DiffTwo(\skipr;\stmt, \hstmt)$ \Comment{Non-deterministic $\skipr$ \  introduction}
\ElsIf {*} \label{line:nondet3}
\Return $\DiffTwo(\stmt, \hstmt;\skipr)$ \Comment{Non-deterministic $\skipr$ \ introduction}
\ElsIf {*} \label{line:nondet4}
\Return $\DiffTwo(\stmt, \skipr; \hstmt)$ \Comment{Non-deterministic $\skipr$ \ introduction}
\ElsIf {$\stmt = \stmt_1 ; \stmt_2$ and $\hstmt = \hstmt_1 ; \hstmt_2$}
\State $(\hstmt_1', \edit_1, \hedit_1) := \DiffTwo(\stmt_1, \hstmt_1)$
\State $(\hstmt_2', \edit_2, \hedit_2) := \DiffTwo(\stmt_2, \hstmt_2)$
\State \Return $(\hstmt_1' ; \hstmt_2', \edit_1 :: \edit_2, \hedit_1 :: \hedit_2)$
\ElsIf {$\stmt = \ifstmt{C}{\stmt_1}{\stmt_2}$ and $\hstmt = \ifstmt{C'}{\hstmt_1}{\hstmt_2}$ and $C = C'$}
\State $(\hstmt_1', \edit_1, \hedit_1) := \DiffTwo(\stmt_1, \hstmt_1)$
\State $(\hstmt_2', \edit_2, \hedit_2) := \DiffTwo(\stmt_2, \hstmt_2)$
\State \Return $(\ifstmt{C}{\hstmt_1'}{\hstmt_2'}, \edit_1 :: \edit_2, \hedit_1 :: \hedit_2)$
\ElsIf {$\stmt = \while{C}{\stmt_1}$ and $\hstmt = \while{C'}{\hstmt_1}$ and $C = C'$}
\State $(\hstmt_1', \edit_1, \hedit_1) := \DiffTwo(\stmt_1, \hstmt_1)$
\State \Return $(\while{C}{\hstmt_1'}, \edit_1 , \hedit_1)$
\Else  \label{line:nondet5}
\State \Return ($\hole, [\stmt], [\hstmt]$)
\EndIf 
\EndProcedure
\end{algorithmic}
\end{algorithm}

Algorithm~\ref{alg:twoDiff} describes $\DiffTwo$ algorithm for computing the 2-way diff.
It takes as input a program $\stmt$ and a program with holes $\hstmt$ and returns the shared program with holes $\hstmt'$ and edits $\edit$ and $\hedit$, such that $\hstmt'[\edit] = \stmt$ and $\hstmt'[\hedit] = \hstmt$.
Since $\hstmt$ may contain holes, the edit $\hedit$ may contain holes. 
The algorithm recursively descends down the structure of the two programs and tries to identify the common program and generate respective edits for the differences.
We use non-deterministic conditional to abstract from actual heuristics to match parts of the two ASTs.
For example, when matching $\stmt$ with $\hstmt_1; \hstmt_2$, a heuristic may decide to match $\stmt$ with $\hstmt_1$ and create a shared program $\hole;\hole$ and edits $\edit = [\stmt, \skipr]$, $\hedit = [\hstmt_1,\hstmt_2]$; 
it may also choose to match $\stmt$ with $\hstmt_2$ and create a shared program $\hole;\hole$ and edits $\edit = [\skipr, \stmt]$, $[\hstmt_1,\hstmt_2]$.  
The decision is often based on algorithms based on variants of {\it longest-common-subsequence}~\cite{longest-common-seq-cite}. 
However, these decisions only help maximize the size of the shared program, and do not affect the soundness of the edit generation.
Lines~\ref{line:nondet1} to ~\ref{line:nondet4} allow us to model all such heuristics by non-deterministically inserting $\skipr$ statemnets before or after a statement. 
Line~\ref{line:nondet5} ensures that the diff procedure can always return by constructing the trivial shared program $\hole$ and $\stmt$ and $\hstmt$ as the respective edits.  
Line~\ref{line:matchHole} checks if $\hstmt$ is a hole, then the shared program is a hole $\hole$ and the two edits contain $\stmt$ and $\hstmt$ respectively.
Line~\ref{line:matcheq} is the case when $\stmt$ equals $\hstmt$.
We use $=$ to denote the syntactic equality of the two syntax trees.
The remaining rules are standard and recurse down the AST structure and match the subtrees.


\fi

\end{document}